\documentclass[reqno,a4paper,11pt]{amsart}
\usepackage[utf8]{inputenc}

\calclayout

\usepackage[utf8]{inputenc}
\usepackage{verbatim}
\usepackage[english]{babel}
\usepackage{amsmath}
\makeatletter
\def\thm@space@setup{%
  \thm@preskip=12pt plus 3pt minus 3pt
  \thm@postskip=\thm@preskip % or whatever, if you don't want them to be equal
}
\makeatother

\usepackage{ esint }
\usepackage{amsfonts}
\usepackage{mathtools}
\usepackage{upgreek}
\usepackage{amsthm}
\usepackage{bm}
\usepackage{float}
\usepackage{color}
\usepackage{tikz}
\usetikzlibrary{decorations.pathreplacing,decorations.markings,calc}
\usetikzlibrary{arrows}
\tikzset{
  on each segment/.style={
    decorate,
    decoration={
      show path construction,
      moveto code={},
      lineto code={
        \path [#1]
        (\tikzinputsegmentfirst) -- (\tikzinputsegmentlast);
      },
      curveto code={
        \path [#1] (\tikzinputsegmentfirst)
        .. controls
        (\tikzinputsegmentsupporta) and (\tikzinputsegmentsupportb)
        ..
        (\tikzinputsegmentlast);
      },
      closepath code={
        \path [#1]
        (\tikzinputsegmentfirst) -- (\tikzinputsegmentlast);
      },
    },
  },
  mid arrow/.style={postaction={decorate,decoration={
        markings,
        mark=at position .5 with {\arrow[scale=1.5, #1]{stealth}}
      }}},
    lftarrow/.style={postaction={decorate,decoration={
        markings,
        mark=at position .25 with {\arrow[#1]{stealth}}
      }}}
    rgtarrow/.style={postaction={decorate,decoration={
        markings,
        mark=at position .75 with {\arrow[#1]{stealth}}
      }}}
}

\tikzset{
	on each segment/.style={
		decorate,
		decoration={
			show path construction,
			moveto code={},
			lineto code={
				\path [#1]
				(\tikzinputsegmentfirst) -- (\tikzinputsegmentlast);
			},
			curveto code={
				\path [#1] (\tikzinputsegmentfirst)
				.. controls
				(\tikzinputsegmentsupporta) and (\tikzinputsegmentsupportb)
				..
				(\tikzinputsegmentlast);
			},
			closepath code={
				\path [#1]
				(\tikzinputsegmentfirst) -- (\tikzinputsegmentlast);
			},
		},
	},
	rmid arrow/.style={postaction={decorate,decoration={
				markings,
				mark=at position .3 with {\arrowreversed[#1]{stealth}}
	}}},
	vmid arrow/.style 2 args={postaction={decorate,decoration={
				markings,
				mark=at position #2 with {\arrow[#1]{stealth}}
	}}},
	vrmid arrow/.style 2 args={postaction={decorate,decoration={
				markings,
				mark=at position #2 with {\arrowreversed[#1]{stealth}}
	}}},
}

\makeatletter
\def\grd@save@target#1{%
  \def\grd@target{#1}}
\def\grd@save@start#1{%
  \def\grd@start{#1}}
\tikzset{
  grid with coordinates/.style={
    to path={%
      \pgfextra{%
        \edef\grd@@target{(\tikztotarget)}%
        \tikz@scan@one@point\grd@save@target\grd@@target\relax
        \edef\grd@@start{(\tikztostart)}%
        \tikz@scan@one@point\grd@save@start\grd@@start\relax
        \draw[minor help lines] (\tikztostart) grid (\tikztotarget);
        \draw[major help lines] (\tikztostart) grid (\tikztotarget);
        \grd@start
        \pgfmathsetmacro{\grd@xa}{\the\pgf@x/1cm}
        \pgfmathsetmacro{\grd@ya}{\the\pgf@y/1cm}
        \grd@target
        \pgfmathsetmacro{\grd@xb}{\the\pgf@x/1cm}
        \pgfmathsetmacro{\grd@yb}{\the\pgf@y/1cm}
        \pgfmathsetmacro{\grd@xc}{\grd@xa + \pgfkeysvalueof{/tikz/grid with coordinates/major step}}
        \pgfmathsetmacro{\grd@yc}{\grd@ya + \pgfkeysvalueof{/tikz/grid with coordinates/major step}}
        \foreach \x in {\grd@xa,\grd@xc,...,\grd@xb}
        \node[anchor=north] at (\x,\grd@ya) {\pgfmathprintnumber{\x}};
        \foreach \y in {\grd@ya,\grd@yc,...,\grd@yb}
        \node[anchor=east] at (\grd@xa,\y) {\pgfmathprintnumber{\y}};
      }
    }
  },
  minor help lines/.style={
    help lines,
    step=\pgfkeysvalueof{/tikz/grid with coordinates/minor step}
  },
  major help lines/.style={
    help lines,
    line width=\pgfkeysvalueof{/tikz/grid with coordinates/major line width},
    step=\pgfkeysvalueof{/tikz/grid with coordinates/major step}
  },
  grid with coordinates/.cd,
  minor step/.initial=.2,
  major step/.initial=1,
  major line width/.initial=2pt,
}

\usepackage{graphicx}
\usepackage{subcaption}
\usepackage{enumerate}
\usepackage{bm}
\usepackage{seqsplit}

\usepackage{cite}

\mathtoolsset{showonlyrefs}

\usepackage[pdftex,colorlinks]{hyperref}
\hypersetup{
	colorlinks=true,
	linkcolor=blue,
	citecolor=red
}

\makeatletter
\def\l@subsection{\@tocline{2}{0pt}{2.5pc}{5pc}{}}

\DeclareMathOperator{\ai}{Ai}
\DeclareMathOperator{\re}{Re}
\DeclareMathOperator{\im}{Im}
\DeclareMathOperator{\ee}{\rm e}
\DeclareMathOperator{\supp}{supp}

\DeclareMathOperator{\sign}{sign}
\renewcommand{\Re}{\mathop{\rm Re}}
\newcommand{\res}{\mathop{\rm Res}}
\renewcommand{\Im}{\mathop{\rm Im}}

\newcommand{\C}{\mathbb{C}}
\newcommand{\R}{\mathbb{R}}
\newcommand{\Z}{\mathbb{Z}}

\newcommand{\boh}{\mathit{o}}
\newcommand{\Boh}{\mathcal{O}}

\newcommand{\ii}{\mathrm{i}}
\newcommand{\dd}{\mathrm{d}}
\newcommand*{\deff}{\mathrel{\vcenter{\baselineskip0.5ex \lineskiplimit0pt
                     \hbox{\scriptsize.}\hbox{\scriptsize.}}}%
                     =}

\DeclareMathOperator{\Li}{Li}

\renewcommand{\bm}{\mathbf} 
\newcommand{\mcal}{\mathcal}

\newcommand{\msf}{\mathsf}

\newcommand{\wh}{\widehat}
\newcommand{\wt}{\widetilde}

\renewcommand{\sp}{\boldsymbol \sigma}

\newcommand{\Tad}{\msf T}
\newcommand{\tad}{\msf t}
\newcommand{\sad}{\msf s}
\newcommand{\uad}{\msf u}

\usepackage{xcolor}
\definecolor{lightgreen}{RGB}{100,200,100}
\usepackage[backgroundcolor=lightgreen]{todonotes}

\newtheorem{theorem}{Theorem}[section]
\newtheorem{prop}[theorem]{Proposition}
\newtheorem{lemma}[theorem]{Lemma}

\theoremstyle{definition}
\newtheorem{definition}[theorem]{Definition}

\newtheorem{rhp}[theorem]{RHP}
\newtheorem{assumption}[theorem]{Assumption}

\theoremstyle{remark}
\newtheorem{remark}[theorem]{Remark}
\numberwithin{equation}{section}

\newcommand*{\defeq}{\mathrel{\vcenter{\baselineskip0.5ex \lineskiplimit0pt
                     \hbox{\scriptsize.}\hbox{\scriptsize.}}}%
                     =}
\newcommand*{\eqdef}{=\mathrel{\vcenter{\baselineskip0.5ex \lineskiplimit0pt
                     \hbox{\scriptsize.}\hbox{\scriptsize.}}}%
                     }

\newcommand\restr[2]{{% we make the whole thing an ordinary symbol
		\left.\kern-\nulldelimiterspace % automatically resize the bar with \right
		#1 % the function
		\vphantom{\big|} % pretend it's a little taller at normal size
		\right|_{#2} % this is the delimiter
}}

\definecolor{apricot}{rgb}{0.98, 0.81, 0.69}

\dedicatory{To Caio Eduardo Candido, a friend who left us too soon.}

\begin{document}

%%%%%%%%%%%%%%%%%%%%%%%%%%%%%%
\title{Deformations of OP ensembles in a bulk critical scaling}

\author[C.~Candido]{Caio E. Candido}
\address[CC]{Instituto de Ciências Matemáticas e de Computação (ICMC), Universidade de S\~ao Paulo (USP), Brazil.}

\author[V.~Alves]{Victor Alves}
\address[VA]{Instituto de Ciências Matemáticas e de Computação (ICMC), Universidade de S\~ao Paulo (USP), Brazil.}
\email{victorjulio@usp.br}

\author[T.~Chouteau]{Thomas Chouteau}
\address[TC]{Instituto de Ciências Matemáticas e de Computação (ICMC), Universidade de S\~ao Paulo (USP), Brazil.}
\email{thomas.chouteau@usp.br}

\author[C.~Santos]{Charles F. Santos}
\address[CS]{Instituto de Ciências Matemáticas e de Computação (ICMC), Universidade de S\~ao Paulo (USP), Brazil.}
\email{charles.santos@icmc.usp.br}

\author[G.~Silva]{Guilherme L.~F.~Silva}
\address[GS]{Instituto de Ciências Matemáticas e de Computação (ICMC), Universidade de S\~ao Paulo (USP), Brazil.}
\email{silvag@icmc.usp.br}

\date{}

%%%%%%%%%%%%%%%%%%%%%%%%%%%%%

\begin{abstract}
    %We consider deformed orthogonal polynomials ensembles, where the deformation symbol is ``centralized'' in a bulk point of the non-deformed model. This symbol deforms the limiting correlation kernel, becoming the finite temperature sine kernel \cite{ClaeysTarricone24}. We also show how the deformation acts on the sub-leading term of the recurrence coefficients of the sequence of orthogonal polynomials under the deformed weight.
    We study orthogonal polynomial ensembles whose weights are deformations of exponential weights, in the limit of a large number of particles. The deformation symbols we consider affect local fluctuations of the ensemble around a bulk point of the limiting spectrum. We identify the limiting kernel in terms of a solution to an integrable non-local differential equation. This novel kernel is the correlation kernel of a conditional thinned process starting from the Sine point process, and it is also related to a finite temperature deformation of the Sine kernel as recently studied by Claeys and Tarricone. We also unravel the effect of the deformation on the recurrence coefficients of the associated orthogonal polynomials, which display oscillatory behavior even in a one-cut regular situation for the limiting spectrum.
\end{abstract}

%\keywords{}

\vspace*{-1.6cm}

\maketitle

\setcounter{tocdepth}{2} \tableofcontents 

\section{Introduction}

There are many different reasons for studying orthogonal polynomials. The present work lies on two modest ones: the role they play in the calculation of statistics of random particle systems, and in the unraveling of novel solutions to integrable systems.

From a probabilistic perspective, orthogonal polynomial (shortly OP) ensembles consist of random points $x_1,\hdots, x_n\in \R$ with joint distribution of the form
\begin{equation}\label{eq:deform}
\mcal P_n(x_1,\hdots,x_n) \dd x_1\cdots \dd x_n\deff \frac{1}{\msf Z_n}\prod_{1\leq i<j\leq n}|x_i-x_j|^2 \prod_{j=1}^n \omega(x_j)\, \dd x_1\cdots \dd x_n,
\end{equation}
where $\omega$ is the associated weight, and $\msf Z_n$ is the normalization constant -- also known as partition function -- that turns \eqref{eq:deform} into a probability distribution in $\R^n$,
$$
\msf Z_n\deff \int_{\R^n} \prod_{1\leq i<j\leq n}|x_i-x_j|^2 \prod_{j=1}^n \omega(x_j)\, \dd x_1 \cdots \dd x_n.
$$

The name {\it orthogonal polynomial ensemble} (shortly OPE) stems from the fact that statistics for \eqref{eq:deform} may be computed through orthogonal polynomials. With $P_j(x)=x^j+\text{(lower degree terms)}$ being the $j$-th orthogonal polynomial with respect to the measure $\omega(x)\dd x$, uniquely determined by
$$
\int_\R P_j(x)x^k\omega(x) \, \dd x=0,\quad k=0,\cdots, j-1,
$$
and $h_j>0$ the associated norming constant, determined from
$$
\frac{1}{h_j^2}=\int_{\R} P_j(x)^2\omega(x)\dd x,
$$
we construct the kernel
\begin{equation}\label{deff:corrkernel}
\msf K_n(x,y)
     \deff \sqrt{\omega(x)} \sqrt{\omega(y)}\sum_{j=0}^{n-1}h_j^2P_j(x)P_j(y),
\end{equation}
known as the correlation kernel. The kernel $\msf K_n$ is of prominent relevance to the particle system \eqref{eq:deform}: it is known that \eqref{eq:deform} may be represented in {\it determinantal form}, namely
\begin{equation}\label{eq:detform}
\mcal P_n(x_1,\hdots,x_n) =\det\left( \msf K_n(x_i,x_j) \right)_{i,j=1}^n.
\end{equation}
In short, \eqref{eq:detform} means that all info on the particle system \eqref{eq:deform} is encoded in the correlation kernel $\msf K_n$ and, in turn, on the orthogonal polynomial themselves.

As mentioned, our second motivation comes from integrability structures that emerge from OPs in various ways. The integrability we want to explore concerns the unraveling of connections with integrable equations that emerge from the asymptotic analysis of the recurrence coefficients $(\gamma_n), (\beta_n)$ in the three-term recurrence relation for the OPs,
\begin{equation}\label{eq:TTRR}
xP_n(x)=P_{n+1}(x)+\beta_nP_n(x)+\gamma_n^2P_{n-1}(x),\quad \gamma_n>0, \; \beta_n\in \R,
\end{equation}
and, as it will turn out, on the kernel $\msf K_n$ as well.

The class of weights we consider in this manuscript are of the form
\begin{equation}\label{eq:deffweight}
\omega(x)=\omega_n(x)\deff \sigma_n(x) \ee^{-nV(x)},  \quad x\in \R,
\end{equation}
where $V:\R\to \R$ is called the potential of the model, and
\begin{equation}\label{def:sigman}
\sigma_n(x)\deff \frac{1}{1+\ee^{-\sad-n^{2}Q(x)}},\quad x \in \R,
\end{equation}
for a fixed function $Q:\R\to \R$. We consider $\sad \in \R$ as a deformation parameter, and when needed we write $\sigma_n(x)=\sigma_n(x\mid \sad)$, $\omega_n(x)=\omega_n(x\mid \sad)$, $\gamma_n=\gamma_n(\sad)$, $\beta_n=\beta_n(\sad)$, $\msf K_n=\msf K_n(\cdot,\cdot\mid \sad)$ etcetera. Conditions on $Q$ will be placed in a moment, but we anticipate that we will impose that 
\begin{equation}\label{eq:conditionsQintro}
Q(p)=Q'(p)=0 \text{ and }Q''(p)>0, \text{ for a fixed } p\in \R \text{ for which }
\lim_{n\to\infty} \frac{1}{n}\msf K_n(p,p\mid \sad)>0.
\end{equation}
In the language of random matrix theory, this last condition is simply saying that $p$ is a regular bulk point of the limiting spectrum. In the language of OPs and potential theory, it means that the density of the underlying equilibrium measure is strictly positive at $p$. We will elaborate more on these aspects later on.

In the limit $\sad\to +\infty$, we have $\sigma_n\to 1$, the weight $\omega_n(\cdot\mid \sad)$ turns into the exponential weight $\omega_n(x\mid \infty)=\ee^{-nV(x)}$, and we refer to the corresponding point process \eqref{eq:deform} as the {\it ground point process}. Our main goal is to understand how the introduction of the term $\sigma_n$ affects asymptotic properties of the correlation kernel and recurrence coefficients when compared to the ones 
$$
\msf K_n(x,y\mid \infty)\deff \lim_{\sad\to +\infty}\msf K_n(x,y\mid \sad),\quad \gamma_n(\infty)\deff \lim_{\sad\to +\infty}\gamma_n(\sad),\quad 
\beta_n(\infty)\deff \lim_{\sad\to +\infty}\beta_n(\sad),
$$ 
corresponding to the ground process. 

There are various reasons why we choose the factor $\sigma_n$ as in \eqref{def:sigman}.
The particular form of $\sigma_n$ is inspired by the finite temperature deformation factors $(1+\ee^{-x})^{-1}$ that appear in free-fermionic models in finite temperature. In fact, recently it was realized that such type of deformations of determinantal point processes lead to new interesting features, including connections with free-fermionic models \cite{Dean2015}, the Kardar-Parisi-Zhang equation \cite{GhosalSilva22, AmirCorwinQuastel2011}, nonlocal integrable equations \cite{BothnerCafassoTarricone2021, CafassoClaeysRuzza2021}, among others \cite{Johansson2007, Liechty2020}. At the level of the OP ensemble \eqref{eq:deform} itself, the deformed kernel $\msf K_n(\cdot\mid \sad)$ is the correlation kernel for a conditional thinning process from the ground process \cite{ClaeysGlesner2021}, and $\sad$ may be viewed as a strength parameter for this thinning.

The choice of scaling factor $n^2$ is explained by the order of fluctuations of the ground process. Near a point $p$ satisfying \eqref{eq:conditionsQintro}, fluctuations of the ground point process happen at a scale of order $n^{-1}$, meaning that the process induced by a local variable $\zeta\approx n(x-p)$ has fluctuations of order $1$. Under the conditions in \eqref{eq:conditionsQintro}, $n^2 Q(x)\approx Q''(p)\zeta^2$, and we expect that local fluctuations are affected in a nontrivial manner. Our results are essentially showing that such heuristics are true, and quantitatively computing the effect of such perturbations.

More precisely, we probe the effect of the introduction of $\sigma_n$ into the ground process through the asymptotic analysis of the correlation kernel $\msf K_n(\cdot\mid \sad)$ and the recurrence coefficients $(\gamma_n(\sad))$, $(\beta_n(\sad))$. In short, our results say that in the large $n$ limit, the correlation kernel $\msf K_n(\cdot\mid \sad)$ converges to a novel kernel, which is constructed out of a special solution $\msf \Phi$ to a nonlocal nonlinear integrable differential equation. When $\sad \to +\infty$, this kernel converges to the celebrated Sine Kernel, and at the level of point processes our calculations imply that this novel kernel is precisely the correlation kernel of a conditional thinned version of the Sine point process. The function $\msf \Phi$ is oscillatory, and thus may be viewed as a nonlinear deformation of the sine oscillations described by the Sine Kernel.

At the level of the recurrence coefficients, we show that as $n\to\infty$, $\beta_n(\sad)$ and $\gamma_n(\sad)$ have the same limits as their undeformed counterparts $\beta_n(\infty)$ and $\gamma_n(\infty)$, but differ from the latter in the subleading order $\Boh(n^{-1})$. We compute the leading order correction to the differences $n(\beta_n(\sad)-\beta_n(\infty))$ and $n(\gamma_n(\sad)^2-\gamma_n(\infty)^2)$, and it turns out that they display two nontrivial features. The first feature is the appearance of explicit oscillatory terms, and the second feature is the appearance of a nonlinear term, which satisfies a nonlinear integrable PDE itself, and it may be alternatively characterized through a total integral of the function used to construct the limiting correlation kernel. The integrable equations underneath both the limiting kernel and the recurrence coefficients have recently arisen in the context of finite temperature deformations of the Sine kernel, as obtained by Claeys and Tarricone \cite{ClaeysTarricone24}. 

The appearance of the nonlinear PDE in the subleading asymptotics of the recurrence coefficients could be anticipated, and it is the exact form of such term that may be viewed as one of our nontrivial contributions. The explicit oscillatory (in $n$) term that appears therein, though, came to us with some surprise. In previous works, oscillations were coming either because the support of the underlying equilibrium measure was disconnected \cite{bleher_its, DKMVZ1}, or thanks to discontinuities in the weight itself \cite{FoulquieMartinezFinkelshteinSousa2011}. However, in our work here the underlying measure always has a connected support, and the corresponding potential is analytic: the perturbation $\sigma_n$ lives on a local scale and does not change the equilibrium measure of the system.

We now move on to describing our results in detail.

\subsection{Statement of results}\hfill 

To state our main results, let us introduce the basic conditions on $V$ and $Q$ under which we will work on.

\begin{assumption}[Assumptions on the potential]\label{asump:V}
We assume that $V$ is a polynomial of even degree and positive leading coefficient. Furthermore, we assume that its associated equilibrium measure $\dd \mu_V(x) = \phi_V(x)\dd x$ is {\it one-cut regular}, with a regular bulk point at the origin. This means that $\supp\phi_V=[a,b]$, for some $a,b\in \R$, $a<0<b$, and that its density $\phi_V$ takes the form
$$
\phi_V(x) = \frac{1}{\pi}\sqrt{(b-x)(x-a)} q(x),\quad  \text{with }\quad q(x) > 0 \text{ for every } x\in [a,b].
$$
Furthermore, we also assume that the variational inequality in the Euler-Lagrange equations associated to the equilibrium problem for $V$ are strict, we refer the reader to Section~\ref{sec:eqmeasure} for a detailed account of these assumptions on $V$.
\end{assumption}

\begin{assumption}[Assumptions on the deformation]\label{asump:Q}    
We assume that the function $Q:\R\to \R$ extends to an analytic function in a complex neighborhood of $\R$,
\begin{equation}\label{eq:behQorigin}
Q(x)>0 , \qquad x \in \R \setminus \{0\},
\end{equation}
and 
\begin{equation}\label{eq:asyQ}
    Q(0) =Q'(0)=0,\quad \frac{Q''(0)}{2} \eqdef \msf t >0.
\end{equation}
\end{assumption}

Placing conditions on $V$ (or rather on its equilibrium measure $\phi_V(x)\dd x$) while studying (critical) scaling limits in random matrix theory is rather standard. In our present situation, we are interested in a local scaling limit near a regular bulk point $p$ in the limiting spectrum of particles $\supp\phi_V$, and the assumptions that $\supp\phi_V$ is connected and $p=0\in \supp\phi_V$ are placed only for concreteness and simplicity of presentation. 

Conditions on $Q$ are also based on the fact that $p=0$ is a regular bulk point, as mentioned in the introduction. We did not have to restrict to $p=0$ and could instead have considered any other regular bulk point in the limiting spectrum, but for simplicity of presentation we will from now on assume $p=0$. 

Under the conditions we just placed, we have that $\sigma_n(x)\to 1$ except for $x=0$, and $\sigma_n(0)=(1+\ee^{-\sad})^{-1}$ for every $n$. One then naturally expects that
$$
\frac{1}{n}\msf K_n(x,x\mid \sad)\to \phi_V(x),\quad n\to \infty,
$$
for $x\in \R$ pointwise. When $\sigma_n\equiv 1$ this result is standard in OPs theory and random matrix theory \cite{johansson_1998, Saff_book, Martinez-Finkelshtein-lecturenotes-2006}, and in the presence of $\sigma_n$ it follows from our methods in a standard way. In particular, such convergence explains that the factor $\sigma_n$ does not change the large $n$ global behavior of zeros of $P_n$ or, equivalently, it does not change the large $n$ global behavior of particles of the system \eqref{eq:deform}. 

In contrast, our first result concerns asymptotics of the correlation kernel \eqref{deff:corrkernel} and shows that its local behavior at the origin is drastically changed by the presence of the deformation $\sigma_n$. Such asymptotics will be given in terms of
\begin{equation}\label{eq:lambdainftyTdeffintro}
\msf T \defeq 
    \frac{ \phi_V(0)\pi}{\sqrt{\tad}}\qquad \text{and}\qquad 
    \lambda_\infty(\zeta \mid \sad) \defeq \frac{1}{1 + \ee^{- \sad - \msf T^{-2} \zeta^2}}.
\end{equation}

\begin{theorem}[The limiting correlation kernel]\label{thm:CorrelationKernelAsymptotics}
     For every $\sad \in \R$ and every $\varepsilon \in (0,1)$, the convergence
     \begin{equation}\label{eq:convergenceKnintro}
     \frac{1}{n \phi_V(0) \pi} \msf K_n\left( \frac{\zeta}{n \phi_V(0) \pi},\frac{\xi}{ n \phi_V(0) \pi} \mid \sad  \right)=\msf K_\infty(\zeta,\xi)+\Boh \left( \frac{1}{n^{1-\varepsilon}}\right),\quad n\to \infty,
     \end{equation}
towards the limiting kernel
$$
\msf K_\infty(\zeta,\xi)=\msf K_\infty(\zeta,\xi\mid \sad)\deff \frac{\sqrt{\lambda_\infty(\zeta \mid \sad)}\sqrt{\lambda_\infty(\xi\mid \sad )}}{2 \pi \ii (\zeta - \xi)}\left[ \msf \Phi\left(\frac{\zeta}{\msf T}\right)\msf \Phi\left(-\frac{\xi}{\msf T}\right)-\msf \Phi\left(-\frac{\zeta}{\msf T}\right)\msf \Phi\left(\frac{\xi}{\msf T}\right) \right]
$$
holds true uniformly for $\zeta,\xi$ in compacts of $\R$, where the function $\msf \Phi(\zeta) = \msf \Phi(\zeta \mid \sad,\msf T)$ satisfies the non-local nonlinear equation
    \begin{equation}\label{eq:PDEPhiintro}
        \partial_{\msf T} \msf \Phi(\zeta \mid \sad, \msf T)=\ii \zeta\msf\Phi(\zeta \mid \sad, \msf T) +\frac{1}{2\pi\ii}
        \left(\int_\R \msf \Phi(\xi \mid \sad, \msf T)^2 \lambda_\infty'( \Tad \xi \mid \sad)\, \dd\xi\right)\msf \Phi(-\zeta \mid \sad, \msf T),
    \end{equation}
    with asymptotic behavior
    \begin{equation}\label{eq:bcPhiintro}
        \msf \Phi(\zeta) \sim \ee^{\ii \msf T\zeta},\quad \zeta \to \pm\infty,
    \end{equation}
    valid for any $\sad\in \R,\msf T>0$ fixed.

Furthermore, as $\sad\to +\infty$ the convergence
\begin{equation}\label{eq:sinekernellimit}
\msf K_\infty(\zeta,\xi\mid \sad)=\frac{1}{\pi}\msf S(\zeta,\xi)+\Boh(\ee^{-\sad}),\quad \sad\to +\infty,\qquad \msf S(\zeta,\xi)\deff \frac{\sin(\zeta-\xi)}{\zeta-\xi},
\end{equation}
holds true uniformly for $\zeta,\xi$ in compacts of $\R$.
\end{theorem}

Theorem~\ref{thm:CorrelationKernelAsymptotics} is the bulk analogue of the soft edge convergence result in \cite{GhosalSilva22}, where Ghosal and the last author show that a finite-temperature type deformation of an OP ensemble, when critically tuned at a regular soft edge, leads to a kernel described in terms of an integro-differential generalization of the Painlevé II equation. Similar appearances of integro-differential integrable equations in random matrix theory have also been recently observed in non-hermitian random matrix models by Bothner and Little \cite{BothnerLittle24edge, BothnerLittle24bulk}.

Under the same scaling as \eqref{eq:convergenceKnintro}, the convergence of the ground process kernel $\msf K_n(\cdot\mid \infty)$ towards the Sine kernel $\msf S$ is an instance of the celebrated Sine kernel universality \cite{deift_book, pastur_shcherbina_universality, Erdoes2017a}. In the language of point processes, it means that the random particle system determined by the distribution \eqref{eq:deform} for $\omega=\omega(\cdot\mid \infty)$ converges to the random particle system determined by $\msf S$, which is known as the Sine point process. The convergence \eqref{eq:sinekernellimit} is essentially saying that the limits $\sad\to +\infty$ and $n\to +\infty$ commute. 

The deformed kernel $\msf K_n(\cdot\mid \sad)$ may be viewed as the correlation kernel of a conditional thinned particle system constructed from the ground OP ensemble determined by $\msf K_n(\cdot\mid \infty)$ as follows. We start with the ground process \eqref{eq:deform} for $\omega=\omega(\cdot\mid \infty)$, and color each random particle $x_j$ with probability $\sigma_n(x_j\mid \sad)$, leaving the particle uncolored with complementary probability $1-\sigma_n(x_j\mid \sad)$. Now, we create a new conditional process, which is the process of colored particles conditioned that no particle has been left uncolored. 

As shown by Claeys and Glesner, the correlation kernel for this conditional process is precisely $\msf K_n(\cdot\mid \sad)$. As a consequence of recent results by Claeys and the last author of the present paper in \cite[Section~3.1]{ClaeysSilva2024}, the conditional thinned ensemble determined from $\msf K_n(\cdot\mid \sad)$ converges {\it weakly as a point process} towards a conditional thinned process constructed from the Sine process. The methods developed in \cite{ClaeysSilva2024}, however, do not give access to computing the correlation kernel of this limiting process. Theorem~\ref{thm:CorrelationKernelAsymptotics} is thus strengthening this weak convergence to a convergence of the correlation kernels, and furthermore it yields that the correlation kernel of the conditioned thinned Sine process is precisely $\msf K_\infty$ as constructed here.

To our knowledge, the kernel $\msf K_\infty$ is novel, but as mentioned before the function $\msf \Phi$ itself and its characterization \eqref{eq:PDEPhiintro}--\eqref{eq:bcPhiintro} has appeared recently in a work by Claeys and Tarricone \cite{ClaeysTarricone24}. A RHP studied in the latter work is also at the core of our results, and we will elaborate more on this connection in a moment.

Our second main result concerns asymptotics for the recurrence coefficients $\gamma_n^2=\gamma_n(\sad)^2,\beta_n=\beta_n(\sad)$ of the deformed orthogonal polynomials (recall \eqref{eq:TTRR} and the discussion thereafter).

Asymptotics of recurrence coefficients have a long and relevant historical importance, an interest which remains to our days. The excellent monograph \cite{VanAsscheBook18} by Van Assche reviews many of such developments and history (see also the recent survey \cite{VanAsscheSurvey22} by the same author), and \cite{DeiftPiorkowski24, ClarksonJordaanLoureiro, Barhoumi24, bleher_deano_painleve_I, bleher_deano_yattselev_2017} is a very limited list of references that encompasses various aspects of asymptotics of recurrence coefficients for orthogonal polynomials that have been considered recently in the literature.

In the case of classical orthogonal polynomials, such as Hermite and Laguerre, recurrence coefficients are explicit and their asymptotics have been known for more than a century. In the context of exponential weights, one of the seminal outputs of the introduction of the Riemann-Hilbert machinery to OP theory in the late 1990s is precisely towards the asymptotic analysis of recurrence coefficients. Already in the early RHP-OPs works, Deift et al. proved that for exponential weights with one-cut regular equilibrium measure, recurrence coefficients admit a full asymptotic expansion in inverse powers of $n$ \cite{DKMVZ1}. More explicit expressions for the first few terms in this expansion were calculated by Kuijlaars and Tibboel \cite[Theorem~1.1]{Tibboel}, which in the case of the ground process $\sad = +\infty$ reads
\begin{equation}\label{deff:gammabetainfty}
    \begin{split}
    \gamma_n(\infty)^2 &=  \frac{(b-a)^{2}}{16} + \Boh(n^{-2}), \quad n \to \infty;\\
    \beta_n(\infty) &= \frac{b+a}{2} + \frac{1}{2n(b-a)} \left( \frac{1}{q(b)} - \frac{1}{q(a)}\right) + \Boh(n^{-2}), \quad n \to \infty.
    \end{split}
\end{equation}

We emphasize that the results in \cite{Tibboel, DKMVZ1} assume that the weight is exponential and the equilibrium measure is one-cut. Observe that in such a case the subleading term of $\gamma_n^2$ is of order $n^{-2}$, and that the subleading term in $\beta_n$ is a rather explicit function on the equilibrium density. 

For our result on recurrence coefficients, we introduce the quantities
    \[
        \kappa \defeq \pi \mu_V([0,b]) \qquad \text{and}\qquad
        G_0(\sad) \defeq  \int_{-\infty}^\infty \log \left(1 + \ee^{- \sad - x^2} \right) \, \dd x=-\sqrt{\pi}\Li_{3/2}(-\ee^{-\sad}),
    \]
and still use $\msf T$ and $\msf \lambda_\infty$ as defined earlier in \eqref{eq:lambdainftyTdeffintro}.

\begin{theorem}[Asymptotics for the recurrence coefficients] \label{thm:recurrence coeffs}
    
    For every $0 < \varepsilon < 1$, the expansions
\begin{equation}\label{eq:expRCmain}
    \begin{split}
    \gamma_n(\sad)^2 
        &= \gamma_n(\infty)^2 + \frac{1}{n} \frac{\Tad}{\pi \phi_V(0)} \frac{b-a}{2} \msf Q(\sad) \cos \left(2n \kappa \right) + \Boh(n^{-2 + \varepsilon}), \\
    \beta_n(\sad)
        &= \beta_n(\infty) + \frac{1}{n} \frac{\msf T}{\pi \phi_V(0)}\left[ \frac{a+b}{\sqrt{-ab}} \frac{G_0(\sad)}{2 \pi } -  \frac{2\msf Q(\sad)}{b-a} \left[ (a+b) \cos(2n \kappa) + 2 \sqrt{-ab} \sin (2n\kappa)\right]\right] + \Boh(n^{-2 + \varepsilon}),
    \end{split}
\end{equation}
are valid as $n\to \infty$, where $\msf Q(\sad)=\msf Q(\sad , \msf T)$ is given by
\begin{equation}\label{eq:integralreprQ}
\msf Q(\sad , \msf T) 
    \deff - \frac{1}{4\pi\ii } \int_{\R} \msf \Phi(\xi \mid \sad,\msf T)^2 \lambda_\infty'(\Tad \xi) \dd \xi .
\end{equation}
    Furthermore, the asymptotic decay
    \begin{equation}\label{eq:decayG0q}
        \msf Q(\sad) = \Boh(\ee^{- \sad}),  \quad \sad \to +\infty,
    \end{equation}
    holds true for any $\msf T>0$ fixed.
\end{theorem}

By the very definition of the Polylog as a power series, it is immediate that $\msf G_0(\sad)=\Boh(\ee^{-\sad})$ as $\sad\to +\infty$. Thus, from \eqref{eq:decayG0q} we obtain that
$$
\beta_n(\sad)=\beta_n(\infty)+\Boh(\ee^{-\sad}),\quad \gamma_n^2(\sad)=\gamma_n^2(\infty)+\Boh(\ee^{-\sad}),\quad \sad\to +\infty.
$$
In other words, each sequence of deformed recurrence coefficients converges to its correspondent in the ground process as $\sad \to +\infty$.

The subleading terms in \eqref{eq:expRCmain} involve several noteworthy terms. First of all, the parameters $\sad$ and $\msf T$ are the sole terms that depend on the function $Q$ used in the deformation \eqref{def:sigman} (recall also \eqref{eq:asyQ}). 

Second, the term $G_0$ albeit at first mysterious, is actually natural. The deformation $\sigma_n$ does not change the equilibrium measure of the system, and as a consequence we have to account to it in subleading order terms throughout the analysis. As common, such subleading terms involve the construction of a so-called Szegő function, and $G_0$ appears precisely due to this function.

Third, the subleading terms in \eqref{eq:expRCmain} contain purely oscillatory factors $\cos(2n\kappa)$ and $\sin(2n\kappa)$. The appearance of (quasi)-periodic terms in this type of asymptotic expansion is not uncommon. However, and as mentioned earlier, it is usually associated to either weights with discontinuities \cite{FoulquieMartinezFinkelshteinSousa2011} or weights with multi-cut equilibrium measures \cite{DKMVZ1, bleher_its}. Over here, however, we emphasize that our weight is analytic on a neighborhood of the real axis for any $n$ fixed, and the underlying equilibrium measure is always one-cut and independent of the deformation $\sigma_n$. For each $x\in \R\setminus \{0\}$ fixed, we indeed have that $\sigma_n(x)\to 1$ pointwise; however $\sigma_n(0)=1/(1+\ee^{-\sad})$. In other words, we may interpret that $\sigma_n$ introduces a delta modification of the weight in the limit, and Theorem~\ref{thm:recurrence coeffs} is saying that this modification is already strong enough to generate oscillatory terms in the asymptotic expansion.

Last but not least, \eqref{eq:expRCmain} involves the function $\msf Q$, which is given as a total weighted integral of the solution $\msf \Phi$ to \eqref{eq:PDEPhiintro}. Alternatively, the function $\msf Q$ satisfies the PDE
\begin{equation}\label{eq:PDEQintro}
\partial_{\msf T} \left( \frac{ \partial_{\msf T} \partial_{\sad} \msf Q }{ 2\msf Q } \right) =  \partial_\sad (\msf Q^2) + 1,
\end{equation}
which provides a self-standing definition of $\msf Q$ as the solution to this nonlinear equation, which is integrable. This function $\msf Q$ also appeared in \cite{ClaeysTarricone24}. In this direction, the main contribution of Theorem~\ref{thm:recurrence coeffs} is in unraveling how precisely this function contributes to the asymptotics of recurrence coefficients, and in showing that nonlinear integrable systems may appear in this context even when the underlying equilibrium measure is one-cut regular or without the need of double-scaling limits.

From a broader perspective, one may view Theorems~\ref{thm:CorrelationKernelAsymptotics} and \ref{thm:recurrence coeffs} as part of a broader program of describing critical scaling limits of finite-temperature like deformations of orthogonal polynomials, as initiated with critical (regular) edge scalings in \cite{GhosalSilva22}. At the level of finite temperature deformations of the universal limiting point processes of random matrix theory, such program has seen major developments in recent years \cite{CafassoClaeysRuzza2021, Ruzza2024, ClaeysTarricone24, BothnerCafassoTarricone2021, CafassoRuzza23, CafassoPinheiro25, KimuraNavand24}.

At the technical level, our approach relies on the characterization of OPs through a Riemann-Hilbert Problem \cite{Fokas1992} (shortly RHP) and the asymptotic analysis of it through the nonlinear steepest descent method \cite{deift_zhou, deift_book}. Compared with the asymptotic analysis of OPs in the case $\sigma_n\equiv 1$, the main technical difference lies in the need of a local parametrix at the origin, even though the equilibrium measure is assumed to be regular at the origin. This is needed so because $\sigma_n$ blows up as $\Boh(n^2)$ in sectors of any complex neighborhood of the origin.

Unlike cases when $\sigma_n\equiv 1$, the model problem required still depends on $n$ in a nontrivial manner: its jumps are not piecewise constant nor homogeneous and instead depend on a change of variables of $\sigma_n$. Thus, the model problem itself requires a separate asymptotic analysis. The same phenomenon has been observed in \cite{GhosalSilva22}. At the end of the way, we show that this model problem is asymptotically close to another RHP that can be identified with a RHP introduced recently by Claeys and Tarricone \cite{ClaeysTarricone24}. The latter RHP is the one underlying the integrable equations that emerge in our results.

\subsection{Structure of the paper}\hfill 

The remaining sections will be dedicated to obtain the proofs of Theorems~\ref{thm:CorrelationKernelAsymptotics} and \ref{thm:recurrence coeffs}, and the content of each section is as follows. 

In Section \ref{sec:ModelProblem} we introduce the model problem mentioned earlier, which will be required in the asymptotic analysis of OPs, obtain certain properties of it, and carry out its asymptotic analysis needed. In particular, we show that in a certain limit it matches with an RHP from \cite{ClaeysTarricone24}, and in the limit when $\sad\to +\infty$ it matches with a RHP connected with the Sine kernel. The emergence of the integrable PDEs from our main results is also explained in Section~\ref{sec:ModelProblem}.

In Section \ref{sec:RHPAnalysis}, we apply the Deift-Zhou nonlinear descent method to the RHP for OPs. Most of the analysis is standard, and as said the main difference lies in the need of a parametrix at the origin, which is constructed from the model problem of Section~\ref{sec:ModelProblem}.

In Section \ref{sec:spoils} we unwrap the asymptotic analysis performed in Sections~\ref{sec:ModelProblem} and \ref{sec:RHPAnalysis}, ultimately proving Theorems~\ref{thm:CorrelationKernelAsymptotics} and \ref{thm:recurrence coeffs}.

Finally, in Appendix~ \ref{Ap:LaplaceIntegrals} we present asymptotic expansions for some Laplace-type integrals which are used during the RHP analysis of OPs, more specifically in the construction of the global parametrix.

\subsection{About the notation}\label{sec:notation}  
\hfill

We outline some standard notation that will be used for the rest of the paper, mostly without further reference.

    We use $D_r(z_0)$ to denote the disk on the complex plane centered at $z_0$ with radius $r>0$, and $D_r=D_r(0)$ for the particular case when $z_0=0$ is the origin. In general, we use bold capital letters $\bm Y,\bm \Psi$ etc to denote matrix-valued functions. The letters $\varepsilon,\delta,\eta$ always denote positive constants that can be made arbitrarily small but are kept fixed, and we always emphasize when they may depend on external parameters. These small constants may have different values for different occurrences in the text. 

When we write that $x\to \infty$ for some variable $x$, we mean that $x\to +\infty$ when $x$ is real, or $x\to \infty$ along any direction of the complex plane in case $x$ is allowed to assume values in $\C\setminus \R$ as well. These two distinct cases will always be clear from the context and meaning of the variables involved.

We also use the following matrix notation. We denote by $\bm I$ and $\bm 0$ the identity matrix and the null matrix, respectively, and by $\bm E_{ij}$ the $2\times 2$ matrix with $1$ in the $(i,j)$-entry and $0$ in the remaining entries. For convenience, we set 
\begin{equation}\label{def:U0}
 \bm U_0\deff \frac{1}{\sqrt{2}}
\begin{pmatrix}
    1 & \ii \\ \ii & 1
\end{pmatrix}
=
\frac{1}{\sqrt{2}}
\left(\bm I+\ii \bm E_{12}+\ii \bm E_{21}\right),
\end{equation}
and recall that the Pauli matrices are given by
$$
\sp_1\deff 
\begin{pmatrix}
    0 & 1 \\ 1 & 0
\end{pmatrix}=
\bm E_{12}+\bm E_{21}, \quad
\sp_2\deff 
\begin{pmatrix}
    0 & -\ii \\ \ii & 0
\end{pmatrix}=
-\ii\bm E_{12}+\ii \bm E_{21}, \quad
\sp_3\deff 
\begin{pmatrix}
    1 & 0 \\ 0 & -1
\end{pmatrix}=
\bm E_{11}-\bm E_{22}.   
$$

In the course of the Riemann-Hilbert analysis, we will use matrix norm notation. For a matrix-valued function $M:U\subset \C\to \C^{2\times 2}$, we denote
$$
|M(z)|\deff \max_{i,j=1,2} |M_{ij}(z)|,
$$
where  $M_{ij}$ stands for the $(i,j)$-th entry of $M$. 
For $p\in [1,\infty]$ and a curve $\Gamma\subset U$, we also use the corresponding $L^p(\Gamma)=L^p(\Gamma,|\dd x|)$ norm with respect to the arc length measure $|\dd x|$,
$$
\|M\|_{L^p(\Gamma)}\deff \max_{i,j=1,2}\|M_{ij}\|_{L^p(\Gamma)}.
$$
For simplicity, for any $p,q\in [1,+\infty]$ we also denote
$$
\|M\|_{L^p\cap L^q(\Gamma)}\deff \max \{ \|M\|_{L^p(\Gamma)}, \|M\|_{L^q(\Gamma)} \}.
$$

\subsection*{Acknowledgments}\hfill

We thank Promit Ghosal, Andrei Martínez-Finkelshtein, Leslie Molag and Lun Zhang for valuable discussion. 

C.C. declares that this study was financed, in part, by the São Paulo Research Foundation (FAPESP), Brazil. Process Number \#2020/15699-3.

T.C. declares that this study was financed, in part, by the São Paulo Research Foundation (FAPESP), Brazil. Process Numbers \#2023/10533-8 and \#2025/06240-0.

V.A. declares that this study was financed, in part, by the São Paulo Research Foundation (FAPESP), Brazil. Process Numbers \#2020/13183-0 and \#2022/12756-1.

C.S. declares that this study was financed, in part, by grant \#2023 Provost of Inclusion and Belonging, University of São Paulo (USP).

G.S. declares that this study was financed, in part, by the São Paulo Research Foundation (FAPESP), Brazil, Process Numbers \# 2019/16062-1 and \# 2020/02506-2. He also acknowledges partial financial support by the Brazilian National Council for Scientific and Technological Development (CNPq) under Grant \# 306183/2023-4, and by the {\it Programa de Apoio aos Novos Docentes} PRPI-USP.

\section{The model problem}\label{sec:ModelProblem}

In this section, we introduce the model Riemann-Hilbert Problem, which is the main object behind all of our main results. In Subsection \ref{sec:intromodel} we study the model problem and state admissibility condition for its input data. In Subsection \ref{sec:modelprob_particularcase} we study a relevant particular case of the model problem, obtained from an RHP studied in \cite{ClaeysTarricone24} and connected with some integrable differential equations underlying our later results. Subsection \ref{sec:Asymptotics1} shows that this particular problem degenerates into an RHP describing the sine kernel. In Subsection \ref{sec:Asymptotics2} another asymptotic result is presented, proving that the admissible instances of the model problem are well posed -- at least for large values of the degree $n$ -- and converge to our particular problem.

\subsection{Introduction of the model problem}\label{sec:intromodel}

\hfill 

Introduce the contours
$$
\Sigma_0\deff [0,+\infty), \quad \Sigma_{\pm 1}\deff [0,\ee^{\pm \pi\ii/8}\infty),\quad \Sigma_{\pm 2}\deff [0,\ee^{\pm 7\pi \ii/8}),\quad \Sigma_3\deff (-\infty,0],\quad \Sigma_{\bm \Phi}\deff \bigcup_{j=-2}^3 \Sigma_j.
$$
We orient the rays $\Sigma_{\pm 2}$ and $\Sigma_{3}$ from $\infty$ to the origin, and the remaining rays from the origin to $\infty$, see Figure~\ref{fig:contours_PHI}. In the description of the following RHP, we will use a function $\msf H$ which is assumed to be defined on a neighborhood of $\Sigma_{\bm \Phi}$, more conditions on it will be placed in a moment. The reader may want to keep in mind the particular choice
$$
\msf H(\zeta)=\msf H_\infty(\zeta)\deff \uad \zeta^2+\sad,\quad \zeta\in \C,
$$
where $\uad> 0$ and $\sad\in \R$ are viewed as parameters, which will play a substantial role in our paper, and whose corresponding RHP will be studied in Section~\ref{sec:modelprob_particularcase}. In particular, this choice $\msf H_\infty$ also explains our choice of angles for the arcs of $\Sigma_{\bm \Phi}$: they are such that $\re\msf H_\infty(\zeta)\to +\infty$ as $\zeta\to \infty$ along $\Sigma_{\bm\Phi}$.

\begin{rhp}
\label{rhp:Phi}
Seek for a $2\times 2$ matrix-valued function $\bm \Phi:\C\setminus \Sigma_{\bm \Phi}\to \C^{2\times 2}$ with the following properties.
\begin{enumerate}[\rm (i)]
\item $\bm \Phi:\C\setminus \Sigma_{\bm \Phi}\to \C^{2\times 2}$ is analytic;
\item The matrix $\bm \Phi$ has continuous boundary values $\bm \Phi_\pm$ along $\Sigma_{\bm \Phi}\setminus \{0\}$, and they are related by $\bm \Phi_+(\zeta)=\bm \Phi_-(\zeta)\bm J_{\bm \Phi}(\zeta)$, $\zeta\in \Sigma_{\bm \Phi}$, where the jump matrix is defined through the function
\begin{equation}\label{deff:lambdamodelprobl}
\lambda(\zeta) \defeq \frac{1}{1 + \ee^{- \msf H(\zeta)}},
\end{equation}
and it is given by
$$
\bm J_{\bm \Phi}(\zeta)\deff 
\begin{dcases}
    \bm I+\frac{1}{\lambda(\zeta)}\bm E_{21},\quad &\zeta \in \Sigma_{\bm \Phi}\setminus \R, \\ 
    \lambda(\zeta)\bm E_{12}-\frac{1}{\lambda(\zeta)}\bm E_{21}, & \zeta\in \R\setminus \{0\}=(\Sigma_0\cup\Sigma_3)\setminus \{0\}.
\end{dcases}
$$
\item Setting
\begin{equation}\label{eq:deffU}
    \bm U^+\deff \bm I,\quad 
    \bm U^-\deff -\bm E_{12}+\bm E_{21},    
\end{equation}
$\bm \Phi$ behaves as
\begin{equation}\label{eq:rhpPhiasympt}
    \bm \Phi(\zeta)
        = \left(\bm I+ \frac{\bm \Phi_1}{\zeta} + \Boh \left(\zeta^{-2} \right)\right)\bm U^\pm \ee^{\mp \ii \zeta\sp_3}, \quad \zeta\to \infty, \; \pm \im \zeta>0;
\end{equation}

\item $\bm \Phi$ remains bounded as $\zeta\to 0$.
\end{enumerate}
\end{rhp}

\begin{figure}
    \centering
\begin{tikzpicture}[scale=5, >=stealth]

% Define the rays with angle and corresponding Sigma index
\foreach \angle/\index in {
    0/0,
    22.5/1,
    157.5/2,
    180/3,
    -22.5/-1,
    -157.5/-2
} {
    % Draw the ray
    \draw[] (0,0) -- ({cos(\angle)}, {sin(\angle)});

    \ifdim \angle pt>-90pt
        \ifdim \angle pt<90pt
            \draw[mid arrow] (0,0) -- ({cos(\angle)}, {sin(\angle)});
        \else
            \draw[mid arrow] ({cos(\angle)}, {sin(\angle)}) -- (0, 0);
        \fi
    \else
        \draw[mid arrow] ({cos(\angle)}, {sin(\angle)}) -- (0,0);
    \fi

    % Place the Sigma index label
    \node at ({1.1*cos(\angle)}, {1.1*sin(\angle)}) {$\Sigma_{\index}$};

    % Add angle labels
    \ifnum\index=1
        \node[above left] at ({0.5*cos(\angle)}, {0.5*sin(\angle)}) {$\ee^{ \frac{\pi \ii }{8}}$};
    \fi

    \ifnum\index=2
        \node[above right] at ({0.5*cos(\angle)}, {0.5*sin(\angle)}) {$\ee^{ \frac{7\pi \ii }{8}}$};
    \fi

    \ifnum\index=-1
        \node[below left] at ({0.5*cos(\angle)}, {0.5*sin(\angle)}) {$\ee^{ - \frac{\pi \ii }{8}}$};
    \fi

    \ifnum\index=-2
        \node[ below right] at ({0.5*cos(\angle)}, {0.5*sin(\angle)}) {$\ee^{ -\frac{7\pi \ii }{8}}$};
    \fi
}

% Draw angle arc between Sigma_0 (0) and Sigma_{11} (-pi/8)
\path 
    (0,0) coordinate (A)
    ({0.2*cos(0)},{0.2*sin(0)}) coordinate (B)
    ({0.2*cos(-22.5)}, {0.2*sin(-22.5)}) coordinate (C);

% Draw the origin
\filldraw (0,0) circle (0.5pt);

\end{tikzpicture}
    \caption{The set $\Sigma_{\bm \Phi}$ is the union of the rays $\Sigma_j$, $-2 \leq j \leq 3$.}
    \label{fig:contours_PHI}
\end{figure}
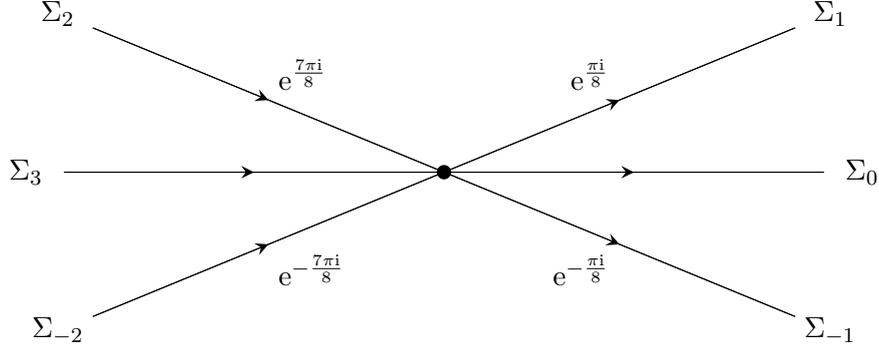

In our case, we are interested in the model problem for a function $\msf H$ arising in a certain structural manner, as we introduce in the next definition.

\begin{definition}\label{deff:admissible}
    We say that a function $\msf H_0$ is {\it admissible} if $\msf H_0$ is analytic on a small fixed disk $D_\delta$ centered at the origin, and if
    $$
    \msf H_0(0)=0,\quad \msf H'_0(0)=0,\quad \uad \defeq  \frac{\msf H''_0(0)}{2}>0.
    $$
\end{definition}

Our interest lies in the RHP~\ref{rhp:Phi} for a particular construction of functions $\msf H_n$, obtained from an admissible function $\msf H_0$ through
\begin{equation}\label{deff:msfHn}
\msf H_n(\zeta)\deff \sad+n^2\msf H_0\left(\frac{\zeta}{n}\right).
\end{equation}
Observe that, this way, $\msf H_n(\zeta)$ is defined only in a growing disk $D_{n\delta}$, and it satisfies
$$
\msf H_n(\zeta)=\sad+\uad \zeta^2 + \Boh\left( \frac{\zeta^3}{n} \right),\quad n\to \infty,
$$
where this expansion is uniform for, say, $|\zeta|<\delta n/2$. In particular, for each $\zeta$ fixed, we have the pointwise convergence
\begin{equation}\label{eq:convHnHinfty}
\msf H_n(\zeta)\to \msf H_\infty(\zeta)\deff \sad +\uad \zeta^2,\quad n\to \infty.
\end{equation}
In particular, we expect that 
\begin{equation}\label{eq:Phi_Definition}
\bm \Phi_n\deff \bm \Phi(\cdot\mid \msf H=\msf H_n)\to \bm \Phi_\infty\deff \bm \Phi(\cdot\mid \msf H=\msf H_\infty),\quad n\to \infty.
\end{equation}
which is what we will prove in a moment. However, to be able to talk about $\bm \Phi_n$ in the first place, we need to be able to say what we mean by the function $\msf H=\msf H_n$ as a function defined on the whole set of rays $\Sigma_{\bm \Phi}$, and not solely on the growing disk $D_{\delta n}$. 

To overcome this raised issue, we need to extend $\msf H_n$ in \eqref{deff:msfHn} from $D_{\delta n}$ to $\Sigma_{\bm \Phi}$. This amounts to extending $\msf H_0(w)$ from $D_{\delta}$ to $\Sigma_{\bm \Phi}$. There are several ways to make this extension, and they will all eventually lead to \eqref{eq:convHnHinfty}. But, for concreteness, we now describe a canonical extension, and we will always work with this extension.

Given $\delta>0$, let $(\phi_j)$ be a partition of unity of $C^\infty_c(\R\to\R)$ real-valued functions, normalized in such a way that
$$
0\leq \phi_j(r)\leq 1\text{ for every } r\in \R,\quad \sum_{j}\phi_j(x)=1 \text{ for } 0\leq r\leq \frac{\delta}{2},\quad \sum_{j} \phi(x)=0, \text{ for } r\notin [0,\delta].
$$
Now, define
\begin{equation}\label{deff:extensionhatH}
\wh{\msf H}_0(w)\deff \left(\sum_j \phi_j(|w|)\right)\msf H_0(w)+\left(1-\sum_{j}\phi_j(|w|)\right)\uad w^2,\quad w\in \C.
\end{equation}
The following lemma is immediate, and we skip its proof.
\begin{lemma}\label{lem:fundineqH}
   The function $\wh{\msf H}_0$ is a $C^\infty$ extension of $\msf H_0$ from $D_{\delta/2}$ to $\C$. Furthermore, for some constants $\eta, M>0$, it satisfies
   $$
   \eta|w|^2\leq \re \wh{\msf H}_0(w)\leq M|w|^2,\quad \text{for every }w\in \Sigma_{\bm\Phi}.
   $$
\end{lemma}

\begin{remark}
     Given an admissible function $\msf H_0$ in the sense of Definition~\ref{deff:admissible}, from now on we already assume that it extends to $\C$ through \eqref{deff:extensionhatH}. Technically speaking, we are extending it from $D_{\delta/2}$ rather than from $D_\delta$, but since $\delta>0$ is an (unimportant) small number, such distinction is irrelevant. From now on, rather than using the notation $\wh{\msf H}_0$ for the extension, to lighten notation we still denote this extension by $\msf H_0$.
\end{remark}

Thus, for an admissible function $\msf H_0$ in the sense of Definition~\ref{deff:admissible}, and for $n> 0$ (not necessarily an integer), we may now set
$$
\msf H_n(\zeta)\deff \sad +n^2\msf H_0\left( \frac{\zeta}{n} \right),\quad  
\quad \lambda_n(\zeta) \defeq \frac{1}{1 + \ee^{- \msf H_n(\zeta)}}, \quad \zeta\in \C,
$$
and talk about the corresponding solution
$$
\bm \Phi_n\deff \bm\Phi(\cdot\mid \msf H=\msf H_n)
$$
of the RHP~\ref{rhp:Phi}. As said before, under the modest conditions we are placing on $\msf H_0$ and hence on $\msf H_n$, there is no apparent reason why the RHP~\ref{rhp:Phi} should be solvable in the first place. But, as also discussed previously (recall \eqref{eq:Phi_Definition}), we will eventually prove that $\bm \Phi_\infty$ exists, and also establish the existence of $\bm\Phi_n$ for large $n$ as a consequence. To verify these claims, we now study $\bm\Phi_\infty$.

\subsection{The model problem: a particular case}\label{sec:modelprob_particularcase}\hfill

As said earlier, a case of RHP~\ref{rhp:Phi} of particular interest comes from the choice
\begin{equation}\label{deff:hinfty}
\msf H_\infty(\zeta)=\msf H_\infty(\zeta\mid \sad)\deff \sad +\uad \zeta^2,\quad \zeta \in \C, \quad \text{leading to}\quad \lambda_\infty(\zeta)\deff \frac{1}{1+\ee^{-\msf H_\infty(\zeta)}},
\end{equation}
and we now analyze the corresponding solution $\bm \Phi_\infty=\bm\Phi(\cdot\mid \msf H=\msf H_\infty)$.

Our goal next is to identify this RHP with the construction from \cite{ClaeysTarricone24}. To that end, let $\mcal S^\pm_k$ be the angular region between $\Sigma_{\pm k}$ and $\Sigma_{\pm (k+1)}$, so that 
\begin{align*}
& \mcal S_0^\pm= \left\{\zeta\in \C \mid 0<\pm \arg \zeta < \frac{\pi}{8} \right\},\quad \mcal S_1^\pm= \left\{\zeta\in \C \mid \frac{\pi}{8}<\pm \arg \zeta < \frac{7\pi}{8} \right\}, \qquad \text{and}\\
& \mcal S_2^\pm = \left\{\zeta\in \C \mid \frac{7\pi}{8}<\pm \arg\zeta <\pi \right\},
\end{align*}
where we recall that all the arguments are taken with the principal branch, that is, $\arg\zeta\in (-\pi,\pi)$.

In \cite{ClaeysTarricone24}, the authors studied Fredholm determinants associated to deformations of the sine kernel, which describe bulk statistics of free fermions at finite temperature. In this direction, they introduce the following Riemann-Hilbert problem.
\begin{rhp} \label{RHP:PsiCT}
Seek for a $2\times 2$ matrix-valued function $\bm \Psi:\C\setminus \R\to \C^{2\times 2}$ with the following properties.
\begin{enumerate}[\rm (i)]
\item $\bm \Psi:\C\setminus \R\to \C^{2\times 2}$ is analytic.
\item The matrix $\bm \Psi$ has continuous boundary values $\bm \Psi_\pm$ along $\R$, and they are related by $\bm \Psi_+(\zeta)=\bm \Psi_-(\zeta)\bm J_{\bm \Psi}(\zeta)$, $\zeta\in \R$, with
$$
\bm J_{\bm \Psi}(\zeta)\deff 
    \bm I+(1-\msf w(\zeta))\bm E_{12},\quad \zeta \in \R.
$$
\item As $\zeta\to \infty$, $\bm \Psi$ behaves as
$$
\bm \Psi(\zeta)=\left(\bm I+ \frac{\bm \Psi_1}{\zeta} + \Boh\left(\frac{1}{\zeta^{2}}\right)\right) \ee^{\ii \msf T\zeta\sp_3} \times 
\begin{dcases}
\begin{pmatrix}
    1 & 1\\ 1 & 0
\end{pmatrix}, & \im\zeta>0, \\ 
\begin{pmatrix}
    1 & 0\\ 1 & -1
\end{pmatrix}, & \im\zeta<0.
\end{dcases}
$$
\end{enumerate}
\end{rhp}

\begin{remark}\label{remark:ClaeysTarricone}
    In \cite{ClaeysTarricone24}, the author use the notation $s$ instead of $\msf T$ to describe the asymptotics of $\bm \Psi$. We change it here for $\msf T$ to avoid confusion, since we already use $\sad$ as a parameter in the definition of $\sigma_n$. Also, the matrix that we denote over here by $\bm \Psi$ is denoted in \cite{ClaeysTarricone24} by $U$, whose RHP is described in Section~4.1 therein.
\end{remark}

    The choice of $\lambda_\infty$ being considered in this subsection corresponds with the choices $\msf W(r)=(1+\ee^{r})^{-1}$ (see Remark 2.4 in \cite{ClaeysTarricone24}) and \begin{equation}\label{eq:deff_w}
        \msf w(\zeta) \defeq \msf W(\zeta^2+\sad) = \frac{1}{1 + \ee^{\sad + \zeta^2}}.
    \end{equation}
        
     They also proved that in such a case $\bm \Psi_1$ has the structure\footnote{These identities are rewritings of the ones obtained in Section~6.1 in \cite{ClaeysTarricone24}, with the identifications $ \msf P=- p$ and $\msf Q=- q$.}
    \begin{equation}\label{eq:Psi1_decomposition}
        \bm \Psi_1 = \ii \msf P(\sad) \sp_3 + \msf Q(\sad) \sp_2.
        \end{equation}
    Observe that $\msf P(\sad)=\msf P(\sad\mid \msf T)$ and $\msf Q(\sad)=\msf Q(\sad\mid \msf T)$, but we often omit the $\msf T$-dependence of these functions, as their $\sad$-dependence is of greater relevance to us.
    
    The functions $\msf P$ and $\msf Q$ are related by $\msf \partial_{\msf T}\msf P=-\msf Q^2$, $\msf Q$ satisfies the identity \eqref{eq:integralreprQ} and solves the PDE \eqref{eq:PDEQintro}, see \cite[$\mathsection$ 6.1]{ClaeysTarricone24}. These functions $\msf P$ and $\msf Q$ will play a role later.

From now on, we make the correspondence
$$
\msf T= \frac{1}{\sqrt{\uad}}
$$
between the variable $\msf T$ of the RHP~\ref{RHP:PsiCT} and the variable $\msf u$ from \eqref{deff:hinfty}. This correspondence is consistent with \eqref{eq:lambdainftyTdeffintro}.

Our objective now is to describe how to match this Riemann-Hilbert problem with the model problem $\bm \Phi_\infty$, i.e. the RHP \ref{rhp:Phi} with $\lambda =\lambda_\infty$. To this end, the identity
\begin{equation}\label{eq:wlambdainfty}
1-\msf w\left(\frac{\zeta}{\msf T} \right)
    = \frac{1}{1+\ee^{-\sad-\uad\zeta^2}} = \lambda_\infty(\zeta).
\end{equation}
heavily motivates the transformation 
$$
\wt {\bm{\Psi}}(\zeta)\deff \bm \Psi(-\zeta/\msf T)\sp_3.
$$
Having in mind that $\zeta\to -\zeta$ changes the $\pm$-boundary values along $\R$ to $\mp$-boundary values, we see that $\wt{\bm \Psi}$ solves the following RHP.
\begin{rhp} Seek for a $2\times 2$ matrix-valued function $\wt{\bm \Psi}:\C\setminus \R\to \C^{2\times 2}$ with the following properties.
\begin{enumerate}[\rm (i)]
\item $\wt{\bm \Psi}:\C\setminus \R\to \C^{2\times 2}$ is analytic.
\item The matrix $\wt{\bm \Psi}$ has continuous boundary values $\wt{\bm \Psi}_\pm$ along $\R$, and they are related by $\wt{\bm \Psi}_+(\zeta)=\wt{\bm \Psi}_-(\zeta)\bm J_{\wt{\bm \Psi}}(\zeta)$, $\zeta\in \R$, with
$$
\bm J_{\wt{\bm \Psi}}(\zeta)\deff 
    \bm I+\lambda_\infty(\zeta)\bm E_{12},\quad \zeta \in \R,
$$
\item As $\zeta\to \infty$, $\wt{\bm \Psi}$ behaves as
$$
\wt{\bm \Psi}(\zeta)=\left(\bm I+\Boh(\zeta^{-1})\right) \ee^{-\ii \zeta\sp_3}
\begin{dcases}
\begin{pmatrix}
    1 & 0\\ 1 & 1
\end{pmatrix}, & \im\zeta>0, \\ 
\begin{pmatrix}
    1 & -1\\ 1 & 0
\end{pmatrix}, & \im\zeta<0.
\end{dcases}
$$

\end{enumerate}
\end{rhp}

Next, we proceed with an opening of lenses, defining
$$
\wh{\bm \Psi}(\zeta)=\wt{\bm \Psi}(\zeta)\times 
\begin{dcases}
    \bm I\mp \frac{1}{\lambda_\infty(\zeta)}\bm E_{21}, & \zeta\in \mcal S_0^\pm\cup \mcal S_2^\pm, \\ 
    \bm I, & \text{elsewhere}.
\end{dcases}
$$

Using extensively that
$$
\ee^{-\ii \zeta \sp_3}\bm U^-=\bm U^- \ee^{\ii \zeta\sp_3},\quad \text{and that}\quad \ee^{\pm \ii\zeta}=\Boh(\ee^{-\eta|\zeta|}) \text{ for } \zeta\to \infty \text{ along }\mcal S^\pm_1,
$$
for some $\eta>0$, we obtain that $\wh{\bm \Psi}$ solves the model problem RHP~\ref{rhp:Phi} for $\lambda=\lambda_\infty$, and therefore
\[
\wh{\bm \Psi}\equiv \bm \Phi_\infty.
\]
Moreover, comparing the coefficients of the asymptotic expression it is straightforward to check that
\begin{equation}\label{eq:Phi1_decomposition}
    \bm \Phi_{\infty,1} = - \msf T \bm \Psi_1 = - \left(\ii \msf p \sp_3 + \msf q \sp_2 \right),
\end{equation}
where for convenience we have set
\begin{equation}\label{eq:msfqpmsfPQ}
\msf p=\msf p(\sad\mid \msf T)\deff \msf T\msf P(\sad\mid \msf T)\quad \text{and}\quad \msf q=\msf q(\sad\mid \msf T)\deff \msf T\msf Q(\sad\mid \msf T).
\end{equation}

We now draw consequences from this identity. Unwrapping the transformations $\wh{\bm \Psi}\mapsto \bm\Psi$, we obtain the identity
\begin{equation}\label{eq:Limiting_Kernelpre}
\left[
\left(\bm I-\frac{\bm E_{21}}{\lambda_\infty(\xi)}\right)\bm\Phi_\infty(\xi)^{-1}\bm\Phi_\infty(\zeta)\left(\bm I+\frac{\bm E_{21}}{\lambda_\infty(\zeta)}\right)
\right]_{21,+}=
    \msf \Phi\left(-\frac{\xi}{\msf T}\right)\msf \Psi\left(-\frac{\zeta}{\msf T}\right)-\msf \Phi\left(-\frac{\zeta}{\msf T}\right)\msf \Psi\left(-\frac{\xi}{\msf T}\right),
\end{equation}
where
\begin{equation}\label{eq:PhiscalarPhimatrix}
\msf \Phi(\xi)\deff \left[\bm\Psi(\xi)\right]_{11,+}=\left[\bm \Phi_{\infty}(-\msf T\xi)\right]_{11,+},\quad
\msf \Psi(\xi)\deff \left[\bm\Psi(\xi)\right]_{21,+}
    =\left[\bm \Phi_{\infty}(-\msf T\xi)\right]_{21,+},\quad \xi\in \R.
\end{equation}
As proven in \cite[Corollary~5.2 and Proposition~5.3]{ClaeysTarricone24},
$$
\msf \Phi(\xi)=\msf \Psi(-\xi),
$$
and we simplify the right-hand side of \eqref{eq:Limiting_Kernelpre} to
\begin{equation}\label{eq:Limiting_Kernel}
\left[
    \left(\bm I-\frac{\bm E_{21}}{\lambda_\infty(\xi)}\right)\bm\Phi_\infty(\xi)^{-1}\bm\Phi_\infty(\zeta)\left(\bm I+\frac{\bm E_{21}}{\lambda_\infty(\zeta)}\right)\right]_{21,+}
        = \msf \Phi\left(\frac{\zeta}{\msf T}\right)\msf \Phi\left(-\frac{\xi}{\msf T}\right)-\msf \Phi\left(-\frac{\zeta}{\msf T}\right)\msf \Phi\left(\frac{\xi}{\msf T}\right).
\end{equation}

By \cite[Theorem 2.1 and Corollary 5.2]{ClaeysTarricone24}, the function $\msf \Phi=\msf \Phi(\zeta \mid \msf T)$ solves the nonlocal nonlinear equation
\begin{equation}\label{eq:nonlocalPDE}
\begin{aligned}
    & \partial_{\msf T} \msf \Phi(\zeta \mid \msf T) =\ii \zeta\msf\Phi(\zeta \mid \msf T ) + \left(\frac{1}{2\pi\ii }\int_\R \msf \Phi(\xi \mid \msf T)^2 \lambda_\infty' \left(\Tad \xi \mid \sad \right)\dd\xi\right)\msf \Phi(-\zeta \mid \msf T), \\
    & \text{with}\quad \msf \Phi(\zeta \mid \Tad) \sim \ee^{\ii \Tad \zeta}, \quad \zeta \to \pm \infty.
\end{aligned}
\end{equation}
We note that we translated the equation from \cite{ClaeysTarricone24} to this equation using the relation \eqref{eq:wlambdainfty}. We stress that $'=\partial_\xi$ is the derivative with respect to the (spectral) variable of the RHP.

\subsection{The model problem: asymptotics I}\label{sec:Asymptotics1}\hfill 

We now analyze $\bm \Phi_\infty(\cdot\mid \sad)$ as $\sad\to +\infty$. As we will see, this RHP converges to the RHP ${\bm \Phi}_{\sin}$ corresponding to the formal choice $\sad =+\infty$, and which gives rise to the celebrated sine kernel.

We point out that $\bm \Phi_{\sin}$ is explicitly constructed by
$$
\bm \Phi_{\sin}(\zeta)\deff 
\begin{dcases}
    \ee^{-\ii\zeta\sp_3}\bm U^\pm (\bm I\pm\bm E_{21}), & \zeta\in \mcal S^\pm_1, \\ 
    \ee^{-\ii\zeta\sp_3}\bm U^\pm , & \zeta\in \mcal S^\pm_0\cup \mcal S^\pm_2. 
\end{dcases}
$$
A direct calculation shows that
\begin{equation}\label{eq:SinekernelRHPcharac}
\frac{1}{2\ii(u-v)}\left[ \left( \bm I-\bm E_{21} \right)\bm \Phi_{\sin}(v)^{-1} {\bm \Phi}_{\sin}(u)\left(\bm I+\bm E_{21}\right) \right]_{21}=\frac{\sin(u-v)}{u-v},
\end{equation}
which justifies the use of the index $\sin$ in ${\bm\Phi}_{\sin}$. With $\msf \Phi_{\sin}$ being defined by \eqref{eq:PhiscalarPhimatrix} with $\bm \Phi_\infty=\bm\Phi_{\sin}$, the identities 
\[
    \msf \Phi_{\sin}(\zeta) = \ee^{\ii \msf T \zeta}\qquad \text{and}\qquad   \partial_{\msf T} \msf \Phi(\zeta \mid \msf T)=\ii \zeta\msf\Phi(\zeta \mid \msf T)
\]
are trivial, so \eqref{eq:Limiting_Kernel} and \eqref{eq:nonlocalPDE} hold with the choice $\lambda_\infty \equiv 1$.

    To verify the claimed convergence, define
\begin{equation}\label{eq:transfLsPhiinfty}
    \bm L_{\sad} (\zeta)\deff \bm \Phi_\infty(\zeta\mid \sad)\bm\Phi_{\sin}(\zeta)^{-1},\quad \zeta\in \C\setminus \Sigma_{\bm \Phi}, \quad
    \lambda_{ \sad}(\zeta) \defeq \frac{1}{1 + \ee^{- \sad - \uad \zeta^2}}.
\end{equation}

\begin{prop}\label{prop:estimatesAsymptotics1}
    The estimates
    \[
        \| \bm L_{\sad, \pm} - \bm I \|_{L^2(\Sigma_{\bm \Phi})} = \Boh(\ee^{- \sad}), \quad
        \| \bm L_{\sad} - \bm I \|_{L^\infty \left( \C \setminus \Sigma_{\bm \Phi}\right)} = \Boh(\ee^{- \sad}),
    \]
    are valid as $\sad \to +\infty$.
\end{prop}

\begin{proof}
    
    Consider the identity
    \[
        \bm J_{\bm L_{\sad}}(\zeta) = \bm \Phi_{\sin,-}(\zeta) \bm J_{\bm \Phi_\infty}(\zeta) \bm J_{\bm \Phi_{\sin}}(\zeta)^{-1} \bm \Phi_{\sin, -}(\zeta)^{-1}.
    \]
    For $\zeta \in \Sigma_{\bm \Phi} \setminus \R$, above takes the form
    \[
        \begin{split}
        \bm J_{\bm L_{\sad}}(\zeta) 
            &= \bm \Phi_{\sin, -}(\zeta) \left( \bm I + \lambda_{\sad}(\zeta)^{-1} \bm E_{21} \right) \left( \bm I - \bm E_{21}\right) \bm \Phi_{\sin, -}(\zeta)^{-1} \\
            &= \bm I + \left( \frac{1}{\lambda_{\sad}(\zeta)} - 1\right) \bm \Phi_{\sin, -}(\zeta) \bm E_{21} \bm \Phi_{\sin, -}(\zeta)^{-1}.
        \end{split}
    \]
    For $\zeta \in \R$, it becomes
    \[
        \begin{split}
        \bm J_{\bm L_{\sad}}(\zeta) 
            &= \bm \Phi_{\sin, -}(\zeta) \left( \lambda_\sad(\zeta)\bm E_{12} - \frac{1}{\lambda_\sad(\zeta) }\bm E_{21}\right) \left( -\bm E_{12} + \bm E_{21} \right) \bm \Phi_{\sin, -}(\zeta)^{-1} \\
            &= \bm I +  \left( \lambda_\sad(\zeta) - 1\right) \bm \Phi_{\sin, -}(\zeta) \bm E_{11} \bm \Phi_{\sin, -}(\zeta)^{-1}+ \left( \frac{1}{\lambda_\sad(\zeta)} - 1\right) \bm \Phi_{\sin, -}(\zeta) \bm{E}_{22} \bm \Phi_{\sin, -}(\zeta)^{-1} \\
            &= \bm I + \left( \frac{1}{\lambda_\sad(\zeta) }-1\right) \left[ \bm \Phi_{\sin, -}(\zeta) \bm{E}_{22} \bm \Phi_{\sin, -}(\zeta)^{-1}  - \lambda_\sad(\zeta) \bm \Phi_{\sin, -}(\zeta) \bm E_{11} \bm \Phi_{\sin, -}(\zeta)^{-1} \right].
        \end{split}
    \]

    Now, from the definition of $\bm \Phi_{\sin}$ we have, for $\pm \Im \zeta > 0$,
    \[
        \begin{split}
        \bm \Phi_{\sin,-}(\zeta) \bm E_{21} \bm \Phi_{\sin,-}(\zeta)^{-1} 
            =  \bm U^\pm \ee^{\mp \ii \zeta \sigma_3} \bm E_{21} \ee^{\pm \ii \zeta \sigma_3} (\bm U^{\pm})^{-1} 
            = \ee^{\pm 2 \ii \zeta}  \bm U^\pm  \bm E_{21}  (\bm U^{\pm})^{-1} .
        \end{split}
    \]
    In particular, $\bm \Phi_{\sin,-}(\zeta) \bm E_{21} \bm \Phi_{\sin,-}(\zeta)^{-1} $ is bounded for $\zeta \in \Sigma_{\bm \Phi} \setminus \R$. For $\zeta \in \R$ we have
    \[
        \bm \Phi_{\sin,-}(\zeta) \bm E_{jj} \bm \Phi_{\sin,-}(\zeta)^{-1} 
            =  \bm U^-  \bm E_{jj}  (\bm U^{-})^{-1}
            =\bm E_{3-j,3-j}, \quad j=1,2.
    \]
     In particular, the previous identities show that $\bm \Phi_{\sin,-}(\zeta) \bm E_{jj} \bm \Phi_{\sin,-}(\zeta)^{-1}$ is bounded for $\zeta \in \R$. Thus, since $|\lambda_{\sad}(\zeta)| \leq 2$ for all $\sad \in \R$ and $\zeta  \in \Sigma_{\bm \Phi}$, there exists $M > 0$, independent of $\sad$, such that
    \begin{equation}\label{eq:estimateJs}
        \| \bm J_{ \bm L_s}(\zeta) - \bm I\| 
            \leq M \ee^{-\sad} |\ee^{- \uad \zeta ^2}|, \quad \zeta \in \Sigma_{\bm \Phi}.
    \end{equation}
    Thus,
    \[
        \| \bm J_{\bm L_{\sad}}(\zeta) - \bm I\|_{L^1  \cap L^\infty(\Sigma_{\bm \Phi})} 
        = \Boh(\ee^{- \sad}), 
        \quad \sad\to +\infty.
    \]
    In other words, the jump of $\bm L_\sad$ converges in $L^\infty$ and $L^1$ (hence in $L^2$) to the identity matrix, and by the well-established small norm theory of RHPs (see for instance \cite[Section~7.5]{deift_book} the result follows.
    \end{proof}

As usual, from the proper decay of the jump matrix to the identity matrix, we now draw consequences for the solution to the RHP itself by a standard application of the small norm theory.

    \begin{theorem}\label{thm:PhiInfinityConv}
    The convergence
        \[
        \bm \Phi_\infty(\zeta \mid \sad) = \bm \Phi_{\sin} (\zeta) \left(\bm I +  \Boh(\ee^{-\sad}) \right), \quad \sad \to +\infty,
    \]
    holds uniformly $\zeta \in \C \setminus \Sigma_{\bm \Phi}$, including for boundary values along $\Sigma_{\bm \Phi}$.

    Furthermore, the estimates
    \begin{equation}\label{eq:estimates_pq}
        \msf p(\sad)= \Boh(\ee^{- \sad}),\quad 
        \msf q(\sad)= \Boh(\ee^{-\sad}), \quad 
        \sad \to +\infty,
    \end{equation}
    hold true, where $\msf p$ and $\msf q$ are as in \eqref{eq:Phi1_decomposition}. 
    
    \end{theorem}

    \begin{proof}
        The convergence of $\bm \Phi_\infty$ is just a rewriting of Proposition \ref{prop:estimatesAsymptotics1} and usual arguments. For the convergence of $\msf p,\msf q$, we start with the perturbative expression
        \begin{align*}
        \bm \Phi_{\infty,1} 
            =-\frac{1}{2\pi\ii}\int_{\Sigma_{\bm \Phi}}(\bm J_{\bm L_\sad}(\zeta)-\bm I)\bm L_{\sad,-}(\zeta)\dd \zeta
            =-\frac{1}{2\pi\ii}\int_{\Sigma_{\bm \Phi}}(\bm J_{\bm L_\sad}(\zeta)-\bm I)\dd \zeta+\Boh(\ee^{-\sad}),
        \end{align*}
        which follows again from the small norm theory of RHPs and \eqref{eq:transfLsPhiinfty}. Now, inequality \eqref{eq:estimateJs} gives
        \[
        |\bm \Phi_{\infty, 1}| 
            \leq \frac{1}{2 \pi} M \ee^{- \sad} \int_{\Sigma_{\bm \Phi}} | \ee^{- \uad \zeta^2} | \, \dd \zeta + \Boh(\ee^{- \sad}).
        \]
        The result now follows from \eqref{eq:Phi1_decomposition}.
    \end{proof}

\subsection{The model problem: asymptotics II}\label{sec:Asymptotics2}
\hfill 

The main goal of this section is to prove Equation \eqref{eq:Phi_Definition}, that is
\[
\bm \Phi_n(\zeta)\to \bm \Phi_\infty(\zeta),\quad n\to \infty.
\]
For this purpose we define
\[
    \bm L_n(\zeta) \defeq \bm \Phi_n(\zeta) \bm \Phi_\infty (\zeta)^{-1},
\]
and, following the small norm theory, we next prove that $\bm L_n\to \bm I$ in the appropriate sense.

Analogously to the approach used in Section \ref{sec:Asymptotics1}, we write
\begin{equation}\label{eq:Asymptotics2_jump_at_R}
    \bm J_{\bm L_n}(\zeta) = \bm \Phi_{\infty, -}(\zeta) \bm J_{\bm \Phi_n}(\zeta) \bm J_{\bm \Phi_\infty}(\zeta)^{-1} \bm \Phi_{\infty, -}(\zeta)^{-1}.
\end{equation}
Expanding the right-hand side we find the identities
\begin{multline}\label{eq:jump_asymptotics2}
    \bm J_{\bm L_n}(\zeta) - \bm I =  \\
    \begin{dcases}
        \left( \frac{1}{\lambda_n(\zeta)} - \frac{1}{\lambda_\infty(\zeta)} \right) \bm \Phi_{\infty, -}(\zeta) \bm E_{21} \bm \Phi_{\infty, -}(\zeta)^{-1}, \quad 
        & \zeta \in \Sigma_{\bm \Phi} \setminus \R, \\
        \left(  \frac{\lambda_n(\zeta)}{\lambda_\infty(\zeta)} - 1\right) \bm \Phi_{\infty, -}(\zeta) \bm E_{11} \bm \Phi_{\infty, -}(\zeta)^{-1} + \left(  \frac{\lambda_\infty(\zeta)}{\lambda_n(\zeta)} - 1\right) \bm \Phi_{\infty, -}(\zeta) \bm E_{22} \bm  \Phi_{\infty, -}(\zeta)^{-1}, \quad
        &\zeta \in \R.
    \end{dcases}
\end{multline}
    Next we establish the appropriate bounds that will allow us to prove that these jumps converge to the identity matrix in the appropriate sense.

\begin{prop}\label{lem:asymptotics2_lem1}
    For every $0 < \varepsilon < 1$ there exists $n_0 \geq 0$ and $M > 0$ such that the inequality
    \[
        \left| \frac{1}{\lambda_n (\zeta)} - \frac{1}{\lambda_\infty(\zeta)}\right|
        \leq M \ee^{- \sad}\frac{ \ee^{- \uad |\zeta|^2}}{n^{\varepsilon}} 
    \]
    holds for every $\zeta\in \Sigma_{\bm \Phi}$ with $|\zeta| \leq n^{\frac{1-\varepsilon}{3}}$, every $\sad\in \R$, and every $n > n_0$.
\end{prop}

\begin{proof}
    Since $\msf H_0$ is analytic on $D_\delta$, 
    \[
        \msf H_0\left( \frac{\zeta}{n} \right) = \uad \zeta^2 + \Boh \left( \frac{ \zeta^3}{n}\right)
    \]
    where the $\Boh$ term is uniform for $| \zeta | \leq  \frac{n \delta}{2}$. In particular, for $|\zeta| < n^{\frac{1-\varepsilon}{3}} < \frac{n \delta}{2}$,
    \[
         \left| \frac{1}{\lambda_n (\zeta)} - \frac{1}{\lambda_\infty(\zeta)}\right|  
            =  \ee^{- \sad} \left| \ee^{-n^2 \msf H_0 \left( \frac{\zeta}{n}\right)} - \ee^{- \uad \zeta^2}\right| 
            = \ee^{-\sad} \left| \ee^{- \uad \zeta^2} \right| \left| \ee^{\Boh \left( n^{-\varepsilon}\right)} - 1\right|  \leq \ee^{- \sad} \ee^{ - \uad |\zeta|^2} \frac{M}{ n^{\varepsilon}}
    \]
    for some $M > 0$, where we used that $0\leq \re (\zeta^2)\leq |\zeta|^2$ for $\zeta\in \Sigma_{\bm \Phi}$.
\end{proof}

\begin{prop}\label{lem:asymptotics2_lem2}
    For some $\eta > 0$, the following estimates holds for $\zeta\in\Sigma_{\bm \Phi}$ with $|\zeta|\geq n^{\frac{1-\varepsilon}{3}}$:
    \[
        \left|\frac{1}{\lambda_n(\zeta)}-\frac{1}{\lambda_\infty(\zeta)}\right|
            \leq 2 \ee^{-\sad} \ee^{-\eta \left(n^{\frac{2(1-\varepsilon)}{3}}+|\zeta^2|\right)}
    \]
\end{prop}
\begin{proof}
    The triangle inequality gives the identity
    \[
        \left|\frac{1}{\lambda_n(\zeta)}-\frac{1}{\lambda_\infty(\zeta)}\right|
            \leq \ee^{-\sad}\left(\left|\ee^{-\uad\zeta^2}\right|+\left|\ee^{-n^2\msf H_0\left(\frac{\zeta}{n}\right)}\right|\right).
    \]
    Now, according to Lemma \ref{lem:fundineqH},
    $$\left|\ee^{-n^2\msf H_0(\zeta/n)}\right|=\ee^{-n^2\Re(\msf H_0(\zeta/n))}\leq\ee^{-\eta|\zeta|^2}.$$
    Then one gets for $\zeta\in\Sigma_{\bm \Phi}$ satisfying $|\zeta|\geq n^{\frac{1 - \varepsilon}{3}}$, 
    \[
        \begin{split}
        \left|\frac{1}{\lambda_n(\zeta)}-\frac{1}{\lambda_\infty(\zeta)}\right|
            &\leq \ee^{-\sad}\left(\ee^{-\uad\Re(\zeta^2)}+\ee^{-\eta|\zeta|^2}\right),\\ 
            &\leq \ee^{-\sad}\left(\ee^{-\uad\alpha|\zeta|^2}\ee^{-\uad\alpha n^{\frac{2(1-\varepsilon)}{3}}}+\ee^{-\frac{\eta}{2}|\zeta|^2}\ee^{-\frac{\eta}{2}n^{\frac{2(1-\varepsilon)}{3}}}\right),
        \end{split}
    \]
    where $\alpha=\frac{1}{2}\cos(\frac{\pi}{4})$.
    The result follows taking  $\eta=\min(\uad\alpha,\eta/2)$.
\end{proof}

\begin{prop}\label{lem:asymptotics2_lem3}
    For $\zeta\in\R$, the inequalities
    $$\left|\frac{\lambda_n(\zeta)}{\lambda_\infty(\zeta)}-1\right|\leq \left|\frac{1}{\lambda_n(\zeta)}-\frac{1}{\lambda_\infty(\zeta)}\right|
    \qquad \text{and}\qquad
    \left|\frac{\lambda_\infty(\zeta)}{\lambda_n(\zeta)}-1\right|\leq \left|\frac{1}{\lambda_n(\zeta)}-\frac{1}{\lambda_\infty(\zeta)}\right|$$
    are valid.
\end{prop}
\begin{proof}
    For $\zeta\in\R$, one has $$\left|\frac{\lambda_n(\zeta)}{\lambda_\infty(\zeta)}-1\right|=\left|\lambda_\infty(\zeta)\right|\left|\frac{1}{\lambda_n(\zeta)}-\frac{1}{\lambda_\infty(\zeta)}\right|\leq \left|\frac{1}{\lambda_n(\zeta)}-\frac{1}{\lambda_\infty(\zeta)}\right|.$$
    The inequality holds since $\lambda_\infty$ is bounded by $1$ on real line. The proof of the second estimate is analogous.
    
\end{proof}

Using the previous propositions, we are now able to apply the small norm theory of RHPs again.

\begin{theorem} \label{Thm:L_n to L_infty}
    For every $0 < \varepsilon < 1$, the estimates
    $$
    \|\bm L_{n,\pm}-\bm I\|_{L^2 (\Sigma_{\bm \Phi})}=\Boh\left(\frac{\ee^{-\sad}}{n^{\varepsilon}}\right) \quad
    \text{and}\quad 
    \|\bm L_{n}-\bm I\|_{L^2 (\C\setminus \Sigma_{\bm \Phi})}=\Boh\left(\frac{\ee^{-\sad}}{n^{\varepsilon}}\right)
    $$
    are valid as $n\to\infty$, for every $\sad \in \R$.
\end{theorem}

\begin{proof}
    The proof consists in using the previous lemmas to study the expression for $\bm L_n - \bm I$ given by Equation \eqref{eq:jump_asymptotics2}. The asymptotics of $\bm \Phi_\infty$ shows that, for $|\zeta|$ sufficiently large, $\zeta \in \Sigma_{\bm \Phi}$, $\pm \Im \zeta > 0$,
\[
    \begin{split}
    \bm \Phi_{\infty, -}(\zeta) \bm E_{21} \bm \Phi_{\infty, -}(\zeta)^{-1} 
    &=  \left( \bm I + \Boh(\zeta^{-1})\right) \bm U^{\pm} \ee^{\mp \ii \zeta \sigma_3} \bm E_{21} \ee^{\pm \ii \zeta \sigma_3} (\bm U^{\pm})^{-1} \left( \bm I + \Boh(\zeta^{-1})\right) \\
    &= \ee^{\pm \ii \zeta}\left( \bm I + \Boh(\zeta^{-1})\right) \bm U^{\pm}  \bm E_{2 1} (\bm U^{\pm})^{-1} \left( \bm I + \Boh(\zeta^{-1})\right)
    \end{split}
\]
    On the other hand, for $\zeta \in \R$, $|\zeta|$ sufficiently large,
    \[
        \begin{split}
            \bm \Phi_{\infty, -}(\zeta) \bm E_{jj} \bm \Phi_{\infty, -}(\zeta)^{-1} 
            &= \left( \bm I + \Boh(\zeta^{-1})\right) \bm U^{\pm}  \bm E_{jj} (\bm U^{\pm})^{-1} \left( \bm I + \Boh(\zeta^{-1})\right)
        \end{split}
    \]
    Identities above and Propositions \ref{lem:asymptotics2_lem1}, \ref{lem:asymptotics2_lem2} and \ref{lem:asymptotics2_lem3} imply that there exists $M > 0$ for which
    \[
        \left| \bm J_{\bm L_n}(\zeta) - \bm I\right| \leq 
        M \ee^{- \sad}\times
        \begin{dcases}
           \frac{\ee^{ - \uad |\zeta|^2}}{n^{\varepsilon}},
            \quad &\zeta \in \Sigma_{\bm \Phi}, |\zeta| \leq n^{\frac{1-\varepsilon}{3}}, \\
              2 \ee^{ - \eta \left(n^{\frac{2(1-\varepsilon)}{3} }+ |\zeta|^2 \right)}
            ,\quad &\zeta \in \Sigma_{\bm \Phi}, |\zeta| \geq n^{\frac{1-\varepsilon}{3}}. \\
        \end{dcases}
    \]
    Thus,
    \begin{equation} \label{Eq:J_ln-I}
        \| \bm J_{\bm L_n} - \bm I \|_{L^1\cap L^\infty(\Sigma_{\bm \Phi})}
            \leq \widehat M \ee^{- \sad } \left[\frac{1}{n^\varepsilon} 
            \left\| \ee^{ - \uad |\zeta|^2} \right\|_{L^1\cap L^\infty(\Sigma_{\bm \Phi})} +
            2 \ee^{- \eta n^{\frac{2(1-\varepsilon)}{3}}} \left\| \ee^{- \eta |\zeta|^2} \right\|_{L^1\cap L^\infty(\Sigma_{\bm \Phi})}
            \right] 
        . 
    \end{equation}
The result now follows from the standard small norm theory for RHPs.
\end{proof}

We are finally able to conclude that the model problem RHP~\ref{rhp:Phi} with admissible data has a solution, and that in fact it is comparable to the one obtained from $\msf H_\infty$ as previously claimed.

\begin{theorem}\label{thm:PhinPhiinfty}
    Fix $\sad_0 \in \R$. For every $0 < \varepsilon < 1$, there exists $n_0 > 0$ such that the solution $\bm \Phi_n$ uniquely exists for $n > n_0$ and $\sad \geq \sad_0$. As $n \to \infty$,
    \begin{equation}\label{eq:asymptotics_phi_n}
        \bm \Phi_n(\zeta \mid \sad) = \bm \Phi_\infty(\zeta \mid \sad) \left( \bm I  + \Boh\left( \frac{\ee^{- \sad}}{n^{\varepsilon}} \right) \right)
        \quad\text{in}\quad L^\infty(\C \setminus \Sigma_{\bm \Phi}),
    \end{equation}
    and the identity extends to boundary values $\bm \Phi_{n,\pm}$ and $\bm \Phi_{\infty,\pm}$ on $\Sigma_{\bm \Phi}$. Moreover,
    \begin{equation} \label{eq:asumptotics_phi_n_1}
        \bm \Phi_{n,1} = \bm \Phi_{\infty,1} + \Boh\left(\frac{\ee^{-\sad}}{n^\varepsilon}\right), \quad n \to +\infty.
    \end{equation}
\end{theorem}

\begin{proof}
    This is proven by arguments similar to the ones made for Theorem \ref{thm:PhiInfinityConv}, this time using Theorem \ref{Thm:L_n to L_infty} and the inequality \eqref{Eq:J_ln-I}. We omit the details for brevity. 
\end{proof}

\begin{remark} \label{rem:bound phi_n}
    For each $n$ fixed, the solution $\bm \Phi_n$ is bounded for $\zeta$ in compacts, and there exists a sufficiently large $R=R_n>0$ such that the asymptotics given in \eqref{eq:rhpPhiasympt} holds true for $|\zeta|\geq R_n$. Now, thanks to the convergence given by Theorem~\ref{thm:PhinPhiinfty}, such $R$ may in fact be chosen to be independent of $n$: simply choose an $R=R_\infty>0$ for which the expansion \eqref{eq:rhpPhiasympt} is valid for $|\zeta|\geq R_\infty$ for the choice $\bm\Phi=\bm\Phi_\infty$.

    As a consequence, we in fact obtain the mild bound
    $$
    \bm \Phi_n(\zeta)=\Boh(1)\ee^{\mp \ii \zeta\sp_3},
    $$
    which is uniform in $n$, $\sad\geq \sad_0$ and also uniform in $\zeta\in \C$, also holding for the boundary values at $\Sigma_{\bm \Phi}$. In particular we see that both $\bm \Phi_n$ and $\bm \Phi_{\infty}$ are bounded in horizontal lines.
\end{remark}

\section{Asymptotic Analysis of the RHP for OPs}\label{sec:RHPAnalysis}

In this section, we move on to the asymptotic analysis of the OPs for the weight \eqref{eq:deffweight}. As usual, this analysis is done using the Deift-Zhou's nonlinear steepest descent method applied to the Fokas-Its-Kitaev RHP characterization of OPs. The main difference compared to the classical undeformed case is the necessity of a local analysis around the origin, due to the accumulation of poles of $\sigma_n$ when $n$ grows large. The construction of this local parametrix will make use of the model problem from Section \ref{sec:ModelProblem}, which will be a key element in the conclusion of the proofs of our main results in Section \ref{sec:spoils}.

\subsection{Equilibrium measures and related quantities}\label{sec:eqmeasure}\hfill

    The first ingredient we will need in the coming analysis is the equilibrium measure for the weight $V$, and we now collect known results about it in the form that will be needed later. Such results are standard \cite{Saff_book, deift_kriecherbauer_mclaughlin}, and here we follow closely the notation and language of \cite{kuijlaars_silva_s_curves}.

The equilibrium measure $\dd \mu_V(x) = \phi_V(x) \, \dd x$ is the unique minimizer to the energy functional
\[
    \iint \log \frac{1}{|x-y|} \, \dd \mu(x) \, \dd \mu(y) + \int V(x) \, \dd \mu(x)
\]
    over all probability measures $\mu$ supported on $\R$ for which $V$ is integrable. For $V$ a polynomial as in Assumption~\ref{asump:V}, this measure uniquely exists, it is absolutely continuous with respect to the Lebesgue measure with a continuous density $\phi_V$, and its support consists of finitely many compact intervals \cite{deift_kriecherbauer_mclaughlin}. The one-cut assumption in Assumption~\ref{asump:V} means that we assume $\mu_V$ to have connected support, say $\supp\mu_V=[a,b]$, and Assumption~\ref{asump:V} also says that $a<0<b$ with $\phi_V(0)>0$.
    
With
$$
U^{\mu_V}(z)\deff \int \log\frac{1}{|z-y|}\dd\mu(y),\quad z\in \C,
$$    
being the logarithmic potential of $\mu_V$, the measure $\mu_V$ is uniquely characterized by the \textit{Euler-Lagrange identities}
    \begin{equation}\label{eq:EulerLagrange}
        \begin{split}
        2 U^{\mu_V}(x) + V(x) &= \ell, \quad x \in [a,b], \\
        2 U^{\mu_V}(x) + V(x) &> \ell, \quad x \in \R \setminus [a,b].
        \end{split}
    \end{equation}
In general, the properties above are understood in the quasi-everywhere sense and with weak inequalities \cite{Saff_book}, but the regularity condition in Assumption~\ref{asump:V} ensures that these properties are true as stated. 

    The complexified logarithmic potential (also known as \textit{$\msf g$-function})
    \begin{equation}\label{eq:g_function}
        \msf g(z) \defeq  \int \log (z-w) \, \dd \mu_V(w)
    \end{equation}
    is an analytic function on $\C \setminus (-\infty, b]$ that admits continuous jumps on $(-\infty,b)$ related by the identities
    \[
        \begin{split}
        \msf g_+(x) - \msf g_-(x) &= 2 \pi \ii  \mu_V([x,b]), \\
        \msf g_+(x) + \msf g_-(x) &= -2 U^{\mu_V}(x).
        \end{split}
    \]
    In particular, the Euler-Lagrange identities 
    % for $\mu_V$ (see Equation~\eqref{eq:EulerLagrange}) 
    imply that
    \[
        \begin{split}
        \msf g_+(x) + \msf g_-(x) - V(x) + \ell &= 0, \quad x \in [a,b], \\
        \msf g_+(x) + \msf g_-(x) - V(x) + \ell &< 0, \quad x \in \R \setminus [a,b].
        \end{split}
    \]

    Let 
    $$
    C^\mu(z) \defeq \int \frac{\dd \mu_V(w)}{w-z},\quad z\in \C\setminus \supp\mu_V,
    $$
    be the Cauchy transform of the measure $\mu_V$. It satisfies the identity
    \begin{equation}\label{eq:CauchyTransform_AlgebraicEq}
        \left( C^{\mu_V}(z) + \frac{V'(z)}{2} \right)^2 = (z-a)(z-b) q(z)^2,\quad z\in \C,
    \end{equation}
    where $q(z)$ is a polynomial of degree $\deg V - 2$ with $q(a), q(b) \neq 0$ and $q(x)>0$ for $x\in(a,b)$. 
    Moreover, Stieltjes Inversion Theorem allows the recovery of the density of $\mu_V$ from the identity above, namely as
    \[
        \phi_V(x) = \frac{1}{\pi} \sqrt{(x-a)(b-x)} q(x), \quad x \in (a,b).
    \]

We will also need certain functions constructed locally from $C^{\mu_V}$, the so-called $\phi$-functions. First of all, the $\phi_b$ is defined by
    \begin{equation}\label{eq:phi_b}
        \phi_b(z) 
            \defeq \int_b^z \left( C^{\mu_V}(x) + \frac{V'(x)}{2}  \right) \, \dd x, \quad z \in \C \setminus (-\infty,b].
    \end{equation}    
    For $x \in (-\infty,b)$, $\phi_b$ admits continuous boundary values 
    \begin{equation} \label{eq:eta_b_pm}
        \phi_{b,\pm}(x) = \mp \int_x^b ((s-a)(s-b))^{\frac{1}{2}}_+ q(s) \, \dd s = \mp \pi \ii  \mu_V([x,b]),
    \end{equation}
    which are related by the following identities
    \begin{equation}\label{eq:phibg}
        \phi_{b,+}(x) + \phi_{b,-}(x) = 0, \quad 2 \phi_{b,\pm}(x) = \mp [ \msf g_+(x) - \msf g_-(x)]
        , \quad x \in (-\infty,b).
    \end{equation}
    Moreover, it is a straightforward calculation to check that
    \begin{equation}\label{eq:phibV}
        2\phi_b(x) = 2U^\mu(x) + V(x) - \ell, \quad x > b,
    \end{equation}
    and that there exists $\varepsilon > 0$ sufficiently small such that for $a < \Re z < b$, $0 < \pm \Im z < \varepsilon$,
    \[
        \Re \phi_b(z) < 0. 
    \]
    
    Similarly, the $\phi_a$ function
    \begin{equation}\label{eq:phi_a}
        \phi_a(z) \defeq \int_a^z \left( C^{\mu_V}(x) + \frac{V'(x)}{2}\right) \, \dd x, \quad z \in \C \setminus (a,+\infty),
    \end{equation}
    satisfies
    \begin{equation}\label{eq:properties_of_phi_a}
        \begin{split}
        2 \phi_a(x) &= 2U^\mu(x) + V(x) - \ell, \quad x < a, \\
        \phi_{a, \pm}(x) &= \phi_{b, \pm}(x) \pm \pi \ii, \quad x \in (a,b), \\
        \phi_a(z) &= \phi_b(z) \pm \pi \ii, \quad \pm \Im z > 0.
        \end{split}
    \end{equation}

    As stated in the introduction, the presence of $\sigma_n$ on the weight introduces a novel local parametrix at $0$ during the asymptotic analysis, whose construction will be done in terms of a conformal map defined in a neighborhood $U_0 \ni 0$, mapping the local parametrix RHP into the model problem (see Section \ref{sec:local_0}). This map $\varphi_0$ is defined as 
    \begin{equation}\label{deff:conformalphi}
        \varphi_0(z) \defeq \mp (\ii \phi_b(z) - \kappa), \quad \pm \Im(z) > 0, \quad \text{ with } \kappa \defeq \ii \phi_{b,+}(0) = \pi \mu_V([0,b]) > 0.
    \end{equation}

    \begin{prop}\label{prop:properties_of_varphi_0}
    The function $\varphi_0$ is a conformal map on a neighborhood of the origin, and
    $$
    \varphi_0(0)=0\quad \text{and}\quad \varphi_0'(0)=\pi\phi_V(0)>0.
    $$
\end{prop}
\begin{proof}
    By construction, $\varphi_0$ is analytic on each component of $U_0\setminus \R$. Equation \eqref{eq:eta_b_pm} then yields the analytic continuation across $\R$ as well. Equation \eqref{eq:CauchyTransform_AlgebraicEq} implies that
    \[
        \varphi_0'(x) = \varphi_{0,+}'(x) = \pi \phi_V(x) >0, \quad x \in U_0 \cap \R,
    \]
    where the positivity claim is ensured by the fact that $\phi_V(0)>0$ and continuity of $\phi_V$ (see Assumption~\ref{asump:V}).
\end{proof}
    
\subsection{The RHP for orthogonal polynomials}\hfill

As mentioned earlier, our starting point for the asymptotic analysis of OPs is their characterization in terms of a RHP due to Fokas, Its and Kitaev \cite{Fokas1992}. Recall that our weight $\omega_n(x)=\sigma_n(x)\ee^{-nV(x)}$ is as in \eqref{eq:deffweight}.
\begin{rhp} \label{rhp:Y}
Seek for a $2\times 2$ matrix-valued function $\bm Y = \bm Y^{(n)}:\C\setminus \R\to \C^{2\times 2}$ with the following properties.
\begin{enumerate}[\rm (i)]
\item $\bm Y:\C\setminus \R\to \C^{2\times 2}$ is analytic.
\item The matrix $\bm Y$ has continuous boundary values $\bm Y_\pm$ along $\R$, and they are related by $\bm Y_+(x)=\bm Y_-(x)\bm J_{\bm Y}(x)$, $x\in \R$, with
$$
\bm J_{\bm Y}(x)\deff \bm I+\omega_n(x)\bm E_{12},\quad x>0.
$$
\item As $z\to \infty$, $\bm Y$ behaves as
$$
\bm Y(z)=\left(\bm I+\frac{\bm Y_1}{z} + 
\frac{\bm Y_2}{z^2} + \Boh(z^{-3})\right)z^{n\sp_3},
$$
where $\bm Y_1,\bm Y_2$ are matrices that depend on $\sad,\tad$ but are independent of $z$.
\end{enumerate}
\end{rhp}

The solution to this RHP depends on $\sigma_n$, and hence on the parameters $\sad,\tad$. In line with the discussion following \eqref{def:sigman} we write $\bm Y=\bm Y(\cdot\mid \sad)=\bm Y(\cdot\mid \tad)=\bm Y(\cdot\mid \sad,\tad)$ when we need to stress this dependence. {\it Ditto} for $\bm Y_1,\bm Y_2$ etc. 

The solution $\bm Y$  encodes the orthogonal polynomials in the following way. The monic orthogonal polynomial $P_n=P_n(\cdot\mid \sad)$ of degree $n$ for the weight $\omega_n$ is obtained as
$$
P_n(z)=(\bm Y(z))_{11}.
$$
Moreover, the correlation kernel introduced in \eqref{deff:corrkernel} for the weight $\omega=\omega_n$ is given by 
\begin{equation}\label{eq:relKnRHPY}
\msf K_n(x,y\mid \sad)= \frac{\sqrt{\omega_n(x\mid \sad)}{\sqrt{\omega_n(y\mid \sad )}}}{2\pi \ii (x-y)}\bm e_2^T \bm Y_+(y\mid \sad)^{-1}\bm Y_+(x\mid \sad)\bm e_1, 
\end{equation}
where $\bm e_1\deff (1,0)^T$ and $\bm e_2\deff (0,1)^T$ are the canonical base vectors for $\R^2$. 

Finally, the recurrence coefficients $\gamma_n^2=\gamma_n(\sad)^2$ and $\beta_n=\beta_n(\sad)$ in \eqref{eq:TTRR} can be written in terms of the RHP~\ref{rhp:Y} as (see for instance \cite[(3.13) and (3.34)]{deift_book}\footnote{In \cite{deift_book}, the author works with recurrence coefficients $(a_n)$ and $(b_n)$ for orthonormal polynomials, not for the monic OPs; the correspondence between such coefficients is standard, and given by $\gamma_n^2=b_{n-1}^2$ and $\beta_n=a_n$.})
    \begin{equation}\label{eq:rec_coef_from_RHP}
        \gamma_n^2 = \left( \bm Y_1^{(n)}\right)_{12} \left( \bm Y_1^{(n)}\right)_{21}, 
        \quad \beta_n = \frac{\left( \bm Y_2^{(n)}\right)_{12}}{\left(\bm Y_{1}^{(n)} \right)_{12}} - \left( \bm Y_1^{(n)}\right)_{22}.
    \end{equation}

We now carry out the steepest descent analysis of $\bm Y$. The transformations involved are now standard in the theory: introduction of the $\msf g$-function constructed out of the equilibrium measure, opening of lenses, construction of global and local parametrices, and final transformation that allows for an application of the perturbation theory of RHPs. 

When compared with the standard RHP analysis for weights with a regular critical measure, there is one major distinction lying at the core of our novel results. Unlike in the classical case, over here we need to construct a local parametrix near the origin, due to the presence of poles of the deformation $\sigma_n$ accumulating at the origin. The construction of this parametrix is novel, and gives rise to all the non-trivial quantities involved in our main results.

\subsection{First transformation: introduction of the \texorpdfstring{$\msf g$}{g}-function}\hfill 

    The first transformation is 
    \[
       \bm T(z) \defeq \ee^{n\frac{\ell }{2} \sp_3} \bm Y(z) \ee^{-n\left(\msf g(z) + \frac{\ell}{2}\right) \sp_3}.
    \]
    where $\msf g$ was introduced in \eqref{eq:g_function} and $\ell$ is the Euler-Lagrange constant (see \eqref{eq:EulerLagrange}).
    This matrix-valued function is the solution to the following RHP.
    \begin{rhp}\label{rhp:T}
        Find a $2 \times 2$ matrix-valued function $\bm T:\C \setminus \R \to \C^{2 \times 2}$ that satisfies:
        \begin{enumerate}[\rm (i)]
            \item $\bm T$ is analytic;
            \item $\bm T$ has continuous boundary values $\bm T_{\pm}$ along $\R$, that are related by the identity $\bm T_+(x) = \bm T_-(x) \bm J_{\bm T}(x)$ where
            \[
                \bm J_{\bm T}(x) = 
                    \ee^{-n[\msf g_+(x) - \msf g_-(x)] \sp_3} + \sigma_n(x) \ee^{n[\msf g_+(x) + \msf g_-(x) + \ell - V(x)]}\bm E_{12}
                % \begin{cases}
                %     \bm I + \sigma_n(x) \ee^{n \left[\msf g_+(x) + \msf g_-(x) + \ell - V(x) \right]} \bm E_{12}, \quad x \in \R \setminus [a,b] \\
                %     \ee^{-n[\msf g_+(x) - \msf g_-(x)] \sp_3} + \sigma_n(x) \bm E_{12}, \quad x \in [a,b]
                % \end{cases}
                ;
            \]
            \item When $z$ becomes unbounded, $\bm T$ behaves as 
            \[
                \bm T(z) = \bm I +   \Boh \left( z^{-1}\right), \quad z\to \infty.
            \]
        \end{enumerate}
    \end{rhp}

    The jump matrix $\bm J_{\bm T}$ can be written in terms of the functions $\phi_a$ and $\phi_b$ that were introduced in \eqref{eq:phi_b} and \eqref{eq:phi_a} as
    \[
        \bm J_{\bm T}(x) = 
            \begin{dcases}
                \bm I + \sigma_n(x) \ee^{-2n \phi_a(x)}\bm E_{12}, \quad &x < a \\
                \begin{pmatrix}
                    \ee^{2n \phi_{b,+}(x)} & \sigma_n(x)  \\
                    0   &   \ee^{2 n \phi_{b,-}(x)}
                \end{pmatrix}, \quad &x \in [a,b] \\
                \bm I + \sigma_n(x) \ee^{-2n \phi_b(x)}\bm E_{12}, \quad &x > b.
            \end{dcases}
    \]

\subsection{Second transformation: opening of lenses}\hfill

    The decomposition
    \[
        \begin{pmatrix}
            \ee^{2n \phi_{b,+}(x)} &   \sigma_n(x) \\
            0   &   \ee^{2n \phi_{b,-}(x)}
        \end{pmatrix}
        =
        \begin{pmatrix}
            1   &   0   \\
            \frac{\ee^{2n \phi_{b,+}(x)}}{\sigma_n(x)}  &   1
        \end{pmatrix}
        \begin{pmatrix}
            0   &   \sigma_n(x) \\
            - \frac{1}{\sigma_n(x)} &   0
        \end{pmatrix}
        \begin{pmatrix}
            1   &   0   \\
            \frac{\ee^{2n \phi_{b,-}(x)}}{\sigma_n(x)} &   1
        \end{pmatrix}
    \]
    motivates the second transformation, which is known as the opening of lenses.

    \begin{figure}[t]
        \centering
        \begin{tikzpicture}[xscale=2,yscale=2]

  % Parameters
  \def\a{-3.1}
  \def\b{3}
  \def\s{-\a/2}
  
  \draw[->] (0,{-\s*sin(45)}) -- (0,{\s*sin(45)});
  % Semirays
  \draw[-, dotted, gray] (0,0) -- ({\s*cos(45)},{\s*sin(45)});
  \draw[-, dotted, gray] (0,0) -- ({-\s*cos(45)},{\s*sin(45)});
  \draw[-, dotted, gray] (0,0) -- ({-\s*cos(45)},{-\s*sin(45)});
  \draw[-, dotted, gray] (0,0) -- ({\s*cos(45)},{-\s*sin(45)});

  % Points a and b
  
  \coordinate (A) at (\a, 0);
  \coordinate (B) at (\b, 0);
  \coordinate (O) at (0,0);
  \coordinate (E1) at ({\a + \a/3},0);
  \coordinate (E2) at ({\b+\b/3},0);

  \draw[->] (E1) -- (E2);

  % Blue segment from a to b on the real axis
  \draw[blue, thick,mid arrow] (A) -- (O);
  \draw[blue, thick,mid arrow] (O) -- (B);
  \draw[ mid arrow] (E1) -- (A);
  \draw[ mid arrow] (B) -- (E2);

  %Lenses
  %Upper lens for a
  \coordinate (Auxap) at ({\a/2}, {-\a/6});
    \draw[dashed, mid arrow] (A) to[out=60, in={180}] (Auxap);
    \draw[dashed, mid arrow] (Auxap) to[out=0, in={180-22.5}] (O);
    \node[above] at (Auxap) {$\Gamma_{\bm S}^+$};
    
    %Lower lens for a
    \coordinate (Auxan) at ({\a/2}, {\a/6});
    \draw[dashed, mid arrow] (A) to[out={360-60}, in={180}] (Auxan);
    \draw[dashed, mid arrow] (Auxan) to[out=0, in={180+22.5}] (O);
    \node[below] at (Auxan) {$\Gamma_{\bm S}^-$};
    
    %Upper lens for b
    \coordinate (Auxbp) at ({\b/2}, {\b/6});
    \draw[dashed, mid arrow] (O) to[out=22.5, in={180}] (Auxbp);
    \draw[dashed, mid arrow] (Auxbp) to[out=0, in={120}] (B);
    \node[above] at (Auxbp) {$\Gamma_{\bm S}^+$};

    %Lower lens for b
    \coordinate (Auxbn) at ({\b/2}, {-\b/6});
    \draw[dashed, mid arrow] (O) to[out=360-22.5, in={180}] (Auxbn);
    \draw[dashed, mid arrow] (Auxbn) to[out=0, in={360-120}] (B);
    \node[below] at (Auxbn) {$\Gamma_{\bm S}^-$};

     % Draw points and labels
  \fill (A) ellipse (0.030cm and 0.030cm);
  \fill (O) ellipse (0.030cm and 0.030cm);
  \fill (B) ellipse (0.030cm and 0.030cm);
  \node[below left] at (A) {$a$};
  % \node[below right] at (O) {$0$};
  \node[below right] at (B) {$b$};
  % \node[anchor=north west] at (1.2,1.2) {$\frac{\pi}{4}$};

    \coordinate (Aux1) at ({\a/2}, 0);
    \coordinate (Aux2) at ({\b/2}, 0);
    \node[below, yshift=-6pt] at (Aux1) {$ \Omega_{\bm S}^-$};
    \node[above, yshift=6pt] at (Aux1) {$\Omega_{\bm S}^+$};
    \node[below, yshift=-6pt] at (Aux2) {$\Omega_{\bm S}^-$};
    \node[above, yshift=6pt] at (Aux2) {$\Omega_{\bm S}^+$};

\end{tikzpicture}
         \caption{Contours for the opening of lenses. $\Gamma_{\bm S}^+$ (resp. $\Gamma_{\bm S}^-$) is the union of the dashed curves in the upper (resp. lower) half plane . We assume that $\Gamma_{\bm S}^\pm$ intersect the imaginary axis only at $z=0$, and at $z=0$ they make an angle $\pi/8$ with the real axis. Also, we assume that they do not intersect the rays with angle equal to $ \pm \frac{ \pi}{4}, \pm \frac{7 \pi}{4}$ as represented in the figure by the dotted gray lines. }
        \label{fig:Opening_of_Lenses}
    \end{figure}
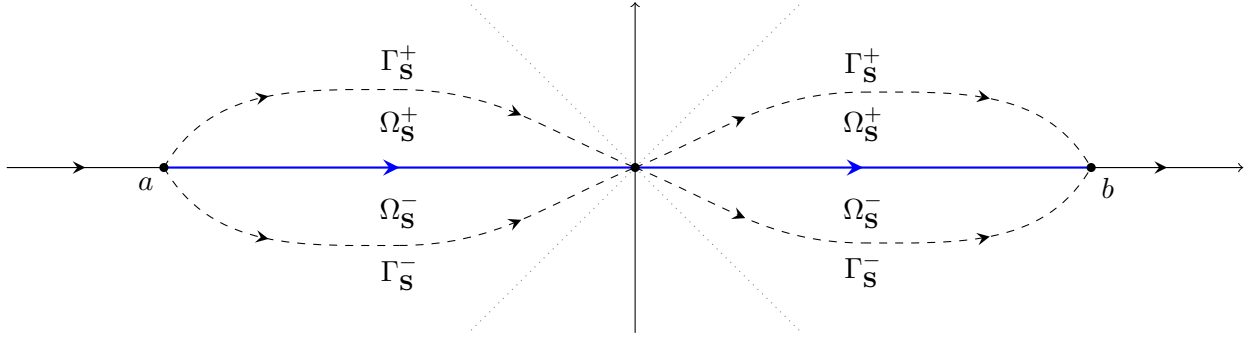

        The lenses are determined by a union of contours $\Gamma_{\bm S} = \Gamma_{\bm S}^- \cup \Gamma_{\bm S}^- \cup \R$, which delimit a union of domains  $\Omega_{\bm S} = \Omega_{\bm S}^+ \cup \Omega_{\bm S}^-$ as in Figure \ref{fig:Opening_of_Lenses}. We emphasize that even though the underlying equilibrium measure is one-cut, we have to open lenses around each of the intervals $(a,0)$ and $(0,b)$ instead of around the full interval $(a,b)$. This is so because along the imaginary axis, the factor $\sigma_n(z)^{-1}$ blows up as $\Boh(\ee^{n^2\eta})$ for some $\eta>0$. Had we chosen to open lenses in the usual way, we would not end up with exponentially decaying jumps near the origin. The next proposition shows that this issue is avoided when the angle between lenses and the real axis is smaller than $\frac{\pi}{4}$ in a neighborhood of $0$.
        
\begin{prop}
\label{prop:polefreeregions}
For every $\varepsilon < \frac{\pi}{4}$, there exists $c(\varepsilon),\delta > 0$ such that $\Re Q(z) \geq c(\varepsilon)$ for every $z$ in the conic regions
$$
\left\{z\in \C \mid |\arg z|<\frac{\pi}{4}-\varepsilon\right\}\cap D_\delta\qquad \text{and}\qquad \left\{z\in \C \mid 0\leq \pi -|\arg z|<\frac{\pi}{4}-\varepsilon\right\}\cap D_\delta.
$$
\end{prop}  

\begin{proof}
    By Assumption \ref{asump:Q},
    \[
        Q(z) = \msf t z^2 + \Boh(z^3), \quad z \to 0.
    \]
    where the $\Boh$ term is uniform for $z$ in compacts that contains $0$. 
    Writing $z = R(\cos(\theta) + \ii \sin(\theta))$ and taking real part, above identity becomes
    \[
        \Re Q(z) = R^2 \left( t \cos(2 \theta) + \Boh(R) \right), \quad z \to 0.
    \]
    Note that $2 \theta \in \left(- \frac{\pi}{2} + 2 \varepsilon,  \frac{\pi}{2} - 2 \varepsilon \right)$ for $\theta = \arg z$ in the enunciated conic regions.  Thus $\msf t\cos (2\theta) \geq c(\varepsilon) > 0$ for $c(\varepsilon) = \cos \left( \frac{\pi}{2} - 2 \varepsilon\right)$. Taking $\delta$ sufficiently small such that for $|z| = R < \delta$ the $\Boh(R)$ term is smaller in norm than $c(\varepsilon)$, the result follows.
\end{proof}
        
The transformation $\bm T\mapsto \bm S$ that performs the opening of lenses is defined as follows (see Figure \ref{fig:Opening_of_Lenses}),
    \[
        \bm S(z) \defeq
            \bm T(z) \times
            \begin{dcases}
                \bm I, \quad &z \in \C \setminus \overline{\Omega}_{\bm S}, \\
                 \bm I \mp \frac{\ee^{2n\phi_b(x)}}{\sigma_n(x)} \bm E_{21}, \quad &z \in \Omega_{\bm S}^\pm.
            \end{dcases}
    \]
We emphasize that we open lenses in such a way that the angle between the lipses of the lenses and the real axis is $\pi/8$. By Proposition~\ref{prop:polefreeregions}, such choice ensures that the factor $1/\sigma_n$ remains bounded and also bounded away from $0$ inside the lenses as $n\to+\infty$.

The  matrix $\bm S$ is the solution for the following RHP.
    \begin{rhp}
        Seek for a matrix-valued function $\bm S : \C \setminus \Gamma_{\bm S} \to \C^{2 \times 2}$ satisfying:
        \begin{enumerate}[\rm (i)]
            \item $\bm S$ is analytic on its domain.
            \item $\bm S$ has continuous boundary values $\bm S_{\pm}$ along $\Gamma_\bm S$, that are related by $\bm S_+(x) = \bm S_-(x) \bm J_{\bm S}(x)$, where
            \begin{equation}\label{deff:JS}
                \bm J_{\bm S}(x) \defeq
                    \begin{dcases}
                        \bm I + \sigma_n(x) \ee^{-2n \phi_a(x)} \bm E_{12}, \quad &x < a, \\
                        \sigma_n(x) \bm E_{12} - \frac{1}{\sigma_n(x)} \bm E_{21}(x), \quad &x \in (a,0) \cup (0,b), \\
                        \bm I + \sigma_n(x) \ee^{-2 n \phi_b(x)} \bm E_{12}, \quad &x > b, \\
                        \bm I + \frac{\ee^{2n \phi_b (x)}}{\sigma_n(x)} \bm E_{21}, \quad &x \in \Gamma_{\bm S}^\pm.
                    \end{dcases}
            \end{equation}
            \item When $z$ becomes unbounded, $\bm S$ behaves as 
            \[
                \bm S(z) = \bm I + \Boh(z^{-1}),\quad z\to \infty.
            \]
        \end{enumerate}
    \end{rhp}

After the transformation $\bm T\mapsto \bm S$, the new jumps become exponentially close to the identity matrix for $z$ away from $[a,b]$. For completeness we state this result rigorously with the next Proposition.

\begin{prop}\label{prop:estimateJS}
For any $\sad_0\in \R$ and any bounded open set $G\subset\C$ with $[a,b]\subset G$, there exists $\eta>0$ such that the estimate
$$
\|\bm J_{\bm S}-\bm I\|_{L^1\cap L^\infty(\Gamma_{\bm S}\setminus G)}=\Boh(\ee^{-\eta n}),\quad n\to \infty,
$$
holds true uniformly for $\sad \geq \sad_0$ and uniformly for $\msf t$ in compacts of $(0,+\infty)$.
\end{prop}

Proposition~\ref{prop:estimateJS} follows from the properties of the functions $\phi_a,\phi_b$ in a standard way, and we skip its proof.

As a consequence of Proposition~\ref{prop:estimateJS}, $\bm J_{\bm S}$ fails to be close to the identity solely near $[a,b]$. However, the jumps on this interval can be accomplished exactly via the construction of the so-called global and local parametrices, as we now discuss. Observe that we did not open lenses near $z=0$, so a local parametrix will also have to be constructed near this point.

\subsection{The global parametrix}\hfill 

The global parametrix is the solution to the RHP obtained from the one for $\bm S$ by neglecting the exponentially small jumps. Concretely, it is the following RHP.

\begin{rhp} Seek for a $2\times 2$ matrix-valued function $\bm G:\C\setminus [a,b]\to \C^{2\times 2}$ with the following properties.
\begin{enumerate}[\rm (i)]
\item $\bm G:\C\setminus [a,b]\to \C^{2\times 2}$ is analytic.
\item The matrix $\bm G$ has continuous boundary values $\bm G_\pm$ along $(a,b)$, and they are related by $\bm G_+(x)=\bm G_-(x)\bm J_{\bm G}(x)$, $a<x<b$, with
$$
\bm J_{\bm G}(x)\deff \sigma_n(x)\bm E_{12}-\frac{1}{\sigma_n(x)}\bm E_{21},\quad a<x<b.
$$
\item As $z\to \infty$, $\bm G$ behaves as
$$
\bm G(z)=\bm I+\Boh(z^{-1}).
$$
\item The entries of the matrix $\bm G$ have square-integrable behavior as $z\to a,b$.
\end{enumerate}
\end{rhp}

The construction of $\bm G$ is standard, see for instance \cite[Section~10.3]{GhosalSilva22}. Define
\begin{equation} \label{eq:h_exponent}
\msf h(z)\deff \frac{((z-a)(z-b))^{1/2}}{2\pi\ii}\int_a^b \frac{\log\sigma_n(x)}{((x-a)(x-b))^{1/2}_+}\frac{\dd x}{x-z},\quad z\in \C\setminus [a,b],
\end{equation}
with the principal branch of the root, and with the term $((z-a)(z-b))^{\frac{1}{2}}$ with branch cut along $[a,b]$.

The function $\msf h$ has continuous boundary values on $(a,b)$ that are related by the jump condition
\[
    \msf h_+(x) + \msf h_-(x) =  \log \sigma_n(x), \quad a < x < b.
\]
Moreover, the expansion
\[ 
    \msf h(z) = \msf h_0 + \frac{\msf h_1}{z} + \Boh(z^{-2}), \quad z \to \infty,
\]
holds true, where
\[
    \begin{split}
        \msf h_0 = \msf h_0(n)
            &\deff - \frac{1}{2 \pi \ii} \int_a^b \frac{\log \sigma_n(x)}{((x-a)(x-b))_+^{\frac{1}{2}}} \, \dd x 
            = \frac{1}{2 \pi } \int_a^b \frac{\log \sigma_n(x)}{\sqrt{(b-x)(x-a)}} \, \dd x , \\
        \msf h_1 = \msf h_1(n) 
            &\deff  \frac{1}{2 \pi } \int_a^b \frac{x \log \sigma_n(x)}{\sqrt{(b-x)(x-a)}} \, \dd x - \frac{a+b}{2} \msf h_0.
    \end{split}
\]

The matrix $\bm G$ is given by
\begin{equation} \label{eq:global_G}
\bm G(z)=\ee^{- \msf h_0 \sigma_3} \bm M(z)\ee^{\msf h(z)\sp_3},\quad z\in \C\setminus [a,b].
\end{equation}
where\begin{equation} \label{eq:M}
\bm M(z)\deff \bm U_0\left( \frac{z-b}{z-a} \right)^{\sp_3/4}\bm U_0^{-1},\quad z\in \C\setminus [a,b].
\end{equation}
with $\bm U_0$ as in Equation \eqref{def:U0}. Introducing
\begin{equation}\label{deff:M1M2}
        \bm M_1 \deff 
        \frac{a-b}{4}\sp_2, \quad 
        \bm M_2 \deff  \frac{(b-a)^2}{32} \bm I - \frac{b^2-a^2}{8} \sp_2,
\end{equation}
the expansion
\begin{equation}\label{eq:asymptotics_M}
    \bm M(z) = \bm I + \frac{\bm M_1}{z} + \frac{\bm M_2}{z^2} + \Boh(z^{-3}), \quad z \to \infty,
\end{equation}
is valid.

Observe that $\msf h=\msf h(\cdot\mid \sad)$, and therefore $\bm G=\bm G(\cdot\mid \sad)$. For later, we need to establish the behavior of the quantities we just introduced as functions of $n$.

\begin{prop}\label{prop:h_n_asymptotics}
    Set 
    \[
    G_0(\sad) \defeq \int_\R \log \left(
    1 + \ee^{-\sad -  x^2} \right) \, \dd x,\quad 
    \widehat{\msf h}_0 \defeq  
        \frac{ G_0(\sad)}{2 \pi \sqrt{|a|b \msf t }}, \quad
        \quad \wh{\msf h}_1\deff -\frac{a+b}{2}\wh{\msf h}_0,\quad
    \widehat {\msf h} (z) \defeq  
        \widehat{\msf h}_0 \frac{((z-a)(z-b))^{\frac{1}{2}}}{z}.
    \]
The estimates
    \begin{equation}\label{eq:estimates_h0_h1}
        \msf h_0(n) 
            = \frac{\widehat{\msf h}_0}{n} + \Boh \left( n^{-3}\right), \quad 
        \msf h_1(n) 
            = \frac{\widehat{\msf h}_1}{n} + \Boh(n^{-3}),\qquad n\to \infty,
    \end{equation}
    as well as
    \begin{equation}\label{eq:estimate_h}
        \msf h(z) =  \frac{ \widehat{\msf h}(z)}{n} + \Boh\left( n^{-3} \right)
    \end{equation}
    are valid, the latter being valid also uniformly for $z$ on compacts of $\C \setminus [a,b]$, and carrying through to boundary values $\msf h_{\pm}(x)$ for $x$ along $\R \setminus \{a,b\}$. Furthermore, for any $\sad_0\in \R$ fixed, these estimates are valid also uniformly for $\sad\geq \sad_0$ and for $\msf t$ in compacts of $(0,+\infty)$.
\end{prop}

\begin{proof}
All the $n$-dependent quantities involved in this proposition come from Laplace-like integrals. The proof of the current proposition is a direct application of Proposition~\ref{prop:Laplace}, where we establish a more systematic asymptotic analysis of such Laplace-like integrals.
    % ,
\end{proof}

As a consequence of Proposition~\ref{prop:h_n_asymptotics}, we see that
\begin{equation}\label{eq:asymptotics_eh}
\ee^{\pm \msf h(z)\sp_3}=\bm I+\Boh(n^{-1}),\quad n\to \infty,
\end{equation}
and with \eqref{eq:global_G} in mind, together with the fact that $\bm M$ is bounded on $\C\setminus \{a,b\}$,
\begin{equation}\label{eq:gpMGasympt}
\bm M(z)=(\bm I+\Boh(n^{-1}))\bm G(z),\quad n\to \infty,
\end{equation}
both valid uniformly for $z$ in compacts of $\C\setminus \{a,b\}$. We will use this estimate later.

\subsection{The local parametrix near edge points}\hfill 

At the edge points $a$ and $b$, the local parametrix is constructed in terms of the Airy RHP, that we write below for the completeness of calculations that will come later. In the statement of the coming RHP, we denote
$$
\Sigma_{\bm A} \defeq \R\cup (\infty\ee^{2\pi \ii /3},0]\cup (\infty \ee^{-2\pi \ii/3}, 0].
$$
Along $\Sigma_{\bm A}$, the ray $[0,+\infty)$ is oriented outwards the origin, and the remaining three rays are oriented towards the origin.

\begin{rhp}\label{rhp:airy}
    Seek $\bm A:\C \setminus \Sigma_{\bm A} \to \C^{2 \times 2}$ satisfying:
    \begin{enumerate}[\rm (i)]
        \item $\bm A: \C \setminus \Sigma_{\bm A} \to \C^{2 \times 2}$ is analytic.
        \item $\bm A$ admits continuous boundary values $\bm A_{\pm}$ on $\Sigma_{\bm A}$, related by the identity $\bm A_+(x) = \bm A_-(x) \bm J_{\bm A}(x)$, $x\in \Sigma_{\bm A}$, where
        \[
            \bm{J}_{\bm A}(x) \deff 
                \begin{dcases}
                    \bm I + \bm E_{12}, \quad &x>0, \\
                    \bm I +\bm E_{21}, \quad &\arg x =  \pm\frac{2 \pi }{3}, \\
                    \bm E_{12} - \bm E_{21}, \quad &x<0.
                \end{dcases}
        \]
        \item As $\zeta \to \infty$,
        \begin{equation}\label{eq:airy_asymptotics}
            \bm A(\zeta) \sim  \zeta^{-\frac{\sp_3}{4}} \bm U_0 \left( \bm I + 
            \sum_{k=1}^\infty \frac{\bm A_k}{\zeta^{\frac{3k}{2}}} \right) \ee^{-\frac{2}{3}\zeta^{\frac{3}{2}} \sp_3},
        \end{equation}
        with coefficients given by\footnote{The coefficients $\bm A_k$ may be computed from \cite[Equation~(7.30)]{deift_exp_weights}. For the record, the matrix $\Psi^\sigma$ therein and our matrix $\bm A$ are related by
        \[
        \bm A(\zeta)=\sqrt{2\pi}\ee^{-\pi \ii /12}\ee^{\pi \ii \sp_3/4}\Psi^\sigma(\zeta).
        \]
        }
        \[
        \msf a_k \defeq   \frac{(-1)^{k+1} \Gamma \left( 3 k + \frac{1}{2}\right)}{36^k k! \Gamma\left( k + \frac{1}{2}\right)(6k-1)} ,\quad 
        \bm A_k \deff 
        \msf a_k \times \begin{dcases}
            \left(\bm I + 6k \sp_2\right), \quad &k \text{ even},  \\
            \left(\sp_3 + 6 \ii \sp_1\right), \quad &k \text{ odd}.
        \end{dcases}
        \]
        \item $\bm A$ remains bounded as $\zeta \to 0$ along $\zeta \in \C \setminus \Sigma_{\bm A}$.
    \end{enumerate}
\end{rhp}

The solution $\bm A$ itself is given in terms of Airy functions in a canonical construction, namely
$$
\bm A(\zeta)\deff 
\begin{pmatrix}
    \ai(\zeta) & -\ee^{4\pi\ii/4} \ai(\ee^{4\pi \ii/3}\zeta) \\
    -\ii \ai'(\zeta) & \ii \ee^{2\pi\ii/4} \ai'(\ee^{4\pi \ii/3}\zeta)
\end{pmatrix}
\times 
\begin{cases}
    \bm I, & 0<\arg \zeta<\frac{2\pi}{3}, \\ 
    (\bm I+\bm E_{21}), & \frac{2\pi}{3}<\arg \zeta<\pi, \\
    (\bm I+\bm E_{12}), & -\frac{2\pi}{3}<\arg \zeta<0, \\ 
    (\bm I+\bm E_{12})(\bm I-\bm E_{21}), & -\pi<\arg\zeta<-\frac{2\pi}{3}. \\
\end{cases}
$$

    For some $\delta>0$, which will be made sufficiently small as needed, let us set
\[
    U_b = U_{b,\delta} \defeq \{z \in \C \mid |z-b| < \delta \}.
\]
The local parametrix at $b$ is the solution of the following RHP.

\begin{rhp}\label{rhp:local_b}
    Find a matrix-valued function $\bm P_b: U_{b, \delta} \setminus \Gamma_{\bm S} \to \C^{2 \times 2}$ that satisfies the following:
    \begin{enumerate}[\rm (i)]
        \item $\bm P_b$ is analytic and admits a continuous extension to $\overline{U_{b, \delta}} \setminus \Gamma_{\bm S}$;
        \item The matrix $\bm P_b$ has continuous boundary values $\bm P_{b,\pm}$ along $U_{b, \delta} \cap \Gamma_{ \bm S}$, related by the identity $\bm P_{b,+}(x) = \bm P_{b,-}(x) \bm J_{\bm P_b}(x)$, where
        \[
            \bm J_{\bm P_b}(x) = \bm J_{\bm S}(x)=
            \begin{dcases}
                \sigma_n(x) \bm E_{12} - \frac{1}{\sigma_n(x) }\bm E_{21}, \quad &x \in (b - \delta, b), \\
                \bm I + \sigma_n(x) \ee^{-2 n \phi_b(x)} \bm E_{12}, \quad &x \in (b, b+\delta), \\
                \bm I + \frac{\ee^{2 n \phi_b(x)}}{\sigma_n(x)} \bm E_{21}, \quad x \in U_{b, \delta} \cap \Gamma_{\bm S}^\pm.
            \end{dcases}
        \]
        \item As $n \to \infty$,
        \begin{equation}\label{eq:asymptotics_P_b}
            \bm P_b(z) = \left( \bm I + \boh(1)\right) \bm G(z)
        \end{equation}
        uniformly on $\partial U_{b, \delta}$.
        \item   $\bm P_b(z)$ remains bounded as $z \to b$.
    \end{enumerate}
\end{rhp}

The solution of this RHP is given in terms of the map
    \[
        % f_n(z) \defeq n^{\frac{2}{3}} \varphi_b(z), \quad 
        % %
        \varphi_b(z) \defeq \left(\frac{3}{2} \phi_b(z)\right)^{\frac{2}{3}}.
    \]
    The construction of $\phi_b$ in terms of the equilibrium measure - which is regular - shows that $\varphi_b$ is conformal and maps $(b-\delta,b)$ to an interval contained in $(-\infty,0)$. We assume that the lens were chosen in such a way that $\varphi_b$ maps $\Gamma_{\pm} \cap U_{b, \delta}$ to a part of the curve $\arg z = \pm \frac{ 2 \pi }{3}$.

    As usual, one maps RHP \ref{rhp:local_b} to the Airy RHP and, after normalization, the solution is
    \[
        \bm P_b(z) =\bm E_b(z) \bm A \left( n^{\frac{2}{3}} \varphi_b(z) \right) \ee^{(n \phi_b(z) - \frac{1}{2} \log \sigma_n(z)) \sp_3}, \quad \bm E_b(z) \defeq  \bm G(z) \ee^{\frac{1}{2} \log \sigma_n(z) \sp_3} \bm U_0^{-1} \left(n^{\frac{2}{3}} \varphi_b(z) \right)^{\frac{\sp_3}{4}}.
    \]

    One can plug the asymptotic of $\bm A$ from Equation~\eqref{eq:airy_asymptotics} in above identity to obtain more precise information than \eqref{eq:asymptotics_P_b}. As $\sigma_n(z) = \bm I + \Boh\left(\ee^{ - n^2 \eta} \right)$, for $z$ near $b$ and some $\eta > 0$, and $\ee^{\msf h \sp_3} =\bm I + \Boh(n^{-1})$, it is straightforward to check that 
    \begin{equation}\label{eq:precise_asymptotics_P_b}
        \bm P_b(z) 
        = \bm G(z) \sigma_n(z)^{\frac{\sp_3}{2}} \left(  \bm I + \Boh( n^{-1})\right) \sigma_n(z)^{-\frac{\sp_3}{2}} 
        = \left( \bm I + \Boh(n^{-1})\right) \bm G(z),
    \end{equation}
    where the $\Boh(n)$ term is uniform for $z \in \partial U_b$, for $\sad \geq \sad_0$ for any $\sad_0\in \R$ fixed, and also uniform for $\tad$ in compacts of $(0,+\infty)$.

    The analysis for the local parametrix at $a$ follows analogously. The local parametrix itself is the solution to the following RHP.
    \begin{rhp}\label{rhp:local_a}
        Find a matrix-valued function $\bm P_a: U_{a,\delta} \setminus \Gamma_{\bm S} \to \C^{2 \times 2}$ satisfying:
            \begin{enumerate}[\rm (i)]
                \item $\bm P_a$ is analytic and admits a continuous extension to $\overline{U_{a,\delta}} \setminus \Gamma_{\bm S}$.
                \item The matrix $\bm P_a$ admits continuous boundary values $\bm P_{\pm}$ along $U_{a,\delta} \cap \Gamma_{\bm S}$, related by $\bm P_{a,+}(x) = \bm P_{a,-}(x) \bm J_{\bm P_a}(x)$, where
                \[
                    \bm J_{\bm P_a}(x) \deff \bm J_{\bm S}(x),\quad x\in U_{a,\delta} \cap \Gamma_{\bm S}\setminus \{a\}.
                \]
                \item As $n \to \infty,$
                \[
                    \bm P_a (z) = \left( \bm I + \boh(1) \right) \bm G(z)
                \]
                uniformly for $z \in \partial U_{a,\delta}$;
                \item $\bm P_a(z)$ remains bounded as $z \to a$.
            \end{enumerate}
    \end{rhp}

    For $\delta>0$ sufficiently small, the map
    \[
        \varphi_a(z) \defeq \left( \frac{3}{2} \phi_a(z)\right)^{\frac{2}{3}},\quad z\in U_{a,\delta},
    \]
    is conformal, and the solution of the RHP \ref{rhp:local_a} is given by
    \begin{align*}
        & \bm P_a(z) = 
            \bm{E}_a(z) \bm A\left( n^{\frac{2}{3}} \varphi_a(z) \right) \ee^{(n \phi_a(z) - \frac{1}{2} \log \sigma_n(z))\sp_3} \sp_3, \qquad \text{with}\\
            & {\bm E}_a(z) \defeq \bm G(z) \ee^{\frac{1}{2} \log \sigma_n(x) \sp_3} \sp_3 \bm U_0^{-1} \left( n^{\frac{2}{3}} \varphi_a(z) \right)^{\frac{\sp_3}{4}}.
    \end{align*}

    The precise asymptotics on $\partial U_{a,\delta}$ is
    \begin{equation}\label{eq:precise_asymptotics_P_a}
        \bm P_a(z) = \left( \bm I + \Boh(n^{-1}) \right) \bm G(z), \quad n \to \infty,
    \end{equation}
    where the $\Boh(n^{-1})$ term is uniform for $z$ on $\partial U_{a,\delta}$, uniform for $\sad \geq \sad_0$ for any $\sad_0\in \R$ fixed, and also uniform for $\tad$ in compacts of $(0,+\infty)$.
    
\subsection{The local parametrix near the origin through the model problem}\label{sec:local_0}\hfill

Fix a neighborhood $U_0$ of the origin, which will be taken appropriately small as needed along the way. Recall that we have not opened lenses around $x=0$, and instead have kept it as a fixed point in the opening of lenses process, see for instance Figure~\ref{fig:Opening_of_Lenses}. In order to cope with the remaining jumps of $\bm S$ near $x=0$ which are not uniformly decaying to the identity, we also need to construct a local parametrix at $x=0$. Concretely, this local parametrix is the solution to the following RHP.
\begin{rhp}\label{rhp:localP0}
    Find a matrix-valued function $\bm P_0: U_0 \setminus \Gamma_{\bm S} \to \C^{2 \times 2}$ that satisfies the following conditions:
    \begin{enumerate}[\rm (i)]
        \item $\bm P_0$ is analytic on $U_0\setminus \Gamma_{\bm S}$ and admits a continuous extension to $\overline{U_0} \setminus \Gamma_{\bm S}$.
        \item The matrix $\bm P_0$ has continuous boundary values $\bm P_{0,\pm}$ along $U_0 \cap \Gamma_{\bm S}$, and they are related by the identity $\bm P_{0,+}(x) = \bm P_{0,-}(x) \bm J_{\bm S}(x)$, where $\bm J_{\bm S}$ is as in \eqref{deff:JS}.
        
        \item As $n \to \infty$,
        \[
            \bm P_0(z) = \left( \bm I + \Boh(n^{-1})\right) \bm G(z),
        \]
        uniformly for $z\in  \partial U_0$, and uniformly for $\sad \geq \sad_0$, for any $\sad_0\in \R$ fixed.
        \item   $\bm P_0(z)$ remains bounded as $z \to 0$.
    \end{enumerate}
\end{rhp}

As usual in RHP literature, we construct the solution to this problem mapping it to a model RHP, in this case the model problem for $\bm \Phi$ from Section~\ref{sec:ModelProblem}. This shall be done with the help of the conformal map $\varphi_0$ defined in \eqref{deff:conformalphi}: we introduce,
\begin{equation}\label{eq:defH_LocalAt0}
\msf H_0(\zeta) 
    \defeq Q \left( \varphi_0^{-1} \left( \frac{\zeta}{n}\right)\right), \quad
\msf H(\zeta) 
    = \msf H_n(\zeta)\defeq \sad + n^2 \msf H_0 \left( \frac{\zeta}{n}\right), \quad
\uad \defeq \frac{\msf t}{\varphi_0'(0)^2}, \quad
\msf H_\infty(\zeta) 
    \defeq \sad + \uad \zeta^2.
\end{equation}
Observe that $\msf H$ is an admissible function in the sense of Definition~\ref{deff:admissible}. Thus, we may consider the corresponding solution to the model problem
\begin{equation}\label{eq:lambdanmodel0}
\bm\Phi_n(\zeta)\deff \bm\Phi(\zeta\mid \msf H=\msf H_n),\quad \text{with the correspondence} \quad \sigma_n(z)=\lambda_n(\zeta), \; \zeta = n\varphi_0(z),
\end{equation}
and we use it to construct $\bm P_0$ as
\begin{equation}\label{eq:def_P0}
\begin{aligned}
    & \bm P_0(z)\deff \bm E(z)\bm\Phi_n(n\varphi_0(z))\ee^{n\phi_b(z)\sp_3},\quad z\in U_0\setminus \Gamma,\quad \text{with} \\ 
    & \bm E(z)\deff \ee^{-\msf h_0\sp_3} \bm M(z)\ee^{\pm \ii n\kappa\sp_3} (\bm U^\pm)^{-1}, \; \pm \im z>0.
\end{aligned}
\end{equation}
    At this point, we assume that the lenses were opened in such a way that $\varphi_0(z)$ maps the lens segments to segments of $\Sigma_{\bm \Phi} \setminus \R$ defined in Section \ref{sec:ModelProblem}. Of course, as $n \to \infty$, the image of the lenses by $n \phi_0(z)$ will be ``filling out'' $\Sigma_{\bm \Phi} \setminus \R$.

    \begin{theorem}
        The matrix-valued function $\bm P_0(z)$ defined in \eqref{eq:def_P0} is a solution to the RHP~\ref{rhp:localP0}.
    \end{theorem}

    \begin{proof}
        First of all, notice that $\bm E$ is analytic near the origin. Indeed, a direct calculation shows that the jump of $\bm E$ across any interval of the form $(-\delta,\delta)\setminus \{0\}$ with $\delta>0$ sufficiently small is
$$
\bm E_-(z)^{-1}\bm E_+(z)=\bm U^-\ee^{\ii n\kappa\sp_3}\bm J_{\bm M}(z)\ee^{\ii n\kappa\sp_3}\bm U^+=\bm U^-(\bm E_{12}-\bm E_{21})=\bm I.
$$
This shows that $\bm E$ has an isolated singularity at $z=0$. From the very definition of $\bm E$ we also know that it remains bounded as $z\to 0$. Thus, $\bm E$ is indeed analytic as claimed.

From this, one concludes that the jump of $\bm P_0$ is given by
\[
    \bm J_{\bm P_0}(z) 
        = \ee^{-n\phi_{b,-}(z)\sp_3}\bm J_{\bm\Phi_n}(n\varphi_0(z))\ee^{n\phi_{b,+}(z)\sp_3} = \bm J_{\bm S}(z),
\] 
where we use that $\phi_{b,-}(z)+\phi_{b,+}(z)=0,\ z\in(a,b)$.
Moreover, for $z\in\partial U_0$ and as $n\to\infty$, we see that $\zeta=n\varphi_0(z)\to \infty$, and from the asymptotics of $\bm \Phi_n$ given in RHP~\ref{rhp:Phi}--(iii) we obtain
$$
\bm P_0(z)=\ee^{-\msf h_0\sp_3}\bm M(z)\ee^{\pm \ii n\kappa\sp_3}(\bm U^\pm)^{-1}(\bm I+\Boh(n^{-1}))  \bm U^\pm \ee^{\mp \ii \zeta\sp_3}\ee^{n\phi_b(z)\sp_3}.
$$
Now, from the definition of $\varphi_0$ in \eqref{deff:conformalphi} we obtain that $\mp \ii \zeta+n\phi_b(z)=\mp \ii n\kappa$, and the identity above updates to
\[ \begin{split}
    \bm P_0(z)
        &=\ee^{-\msf h_0\sp_3}\bm M(z)\ee^{\pm \ii n\kappa\sp_3}(\bm U^\pm)^{-1}(\bm I+\Boh(n^{-1}))  \bm U^\pm \ee^{\mp \ii n\kappa\sp_3} \\
        &=\ee^{-\msf h_0\sp_3}\bm M(z) \left(\bm I+\Boh(n^{-1})\right) 
        = \bm G(z) \ee^{- \msf h(z) \sp_3} \left(\bm I+\Boh(n^{-1})\right),
\end{split}
\]
where in the second step we used that $\bm U^\pm$ and $\ee^{\mp\ii n\kappa\sp_3}$ are bounded. RHP~\ref{rhp:localP0}--(iii) now follows observing that $\bm M$ is bounded too for $z$ near the origin, and also using \eqref{eq:asymptotics_eh}.
    \end{proof}

For later convenience, we state a mild bound of $\bm P_0$ on a full real neighborhood of the origin.

\begin{lemma} \label{lem:bound_P0}
Both $\bm P_{0,+}(x)$ and $\bm P_{0,+}(x)^{-1}$ remain bounded as $n\to \infty$, uniformly for $x \in U_0 \cap \R$ and uniformly for $\sad\geq s_0$, for any $\sad_0\in \R$.
\end{lemma}

\begin{proof}
    We start by using the uniform convergence given by Theorem \ref{thm:PhinPhiinfty} for $\zeta = n \varphi_0(x)$ to write
    \[
        \bm P_{0,+}(x) = \bm E_n(x) \bm \Phi_{\infty,+}( n \varphi_{0}(x)) \ee^{n \phi_{b,+}(z) \sp_3} + \Boh(\ee^{- \sad }n^{- \varepsilon}),
    \]
    where we used that $\bm E_n(z)$ is bounded around $z=0$ and that $\phi_{b,+}(x)$ is purely imaginary for $x \in (-\infty, b)$. 
    
    By Remark \ref{rem:bound phi_n}, $\bm \Phi_{\infty,+}(\zeta)$ is bounded for $ \zeta \in \R$. Since $n \varphi_0(x) \in \R$ for every $x \in U_0 \cap \R$, the result follows.
\end{proof}

\subsection{Final transformation and small norm theory}\hfill 

For the final transformation, let us set
$$
U\deff U_a\cup U_b\cup U_0,\quad \Gamma_{\bm R}\deff (\Gamma_{\bm S}\cup \partial U)\setminus ( U\cup [a,b]),
$$
where we recall that $\Gamma_{\bm S}$ is the jump contour for $\bm S$ (see Figure \ref{fig:Opening_of_Lenses}),
and introduce 
\begin{equation}\label{deff:PPoPaPb}
\bm P(z)\deff 
\begin{dcases}
    \bm P_0(z),& z\in U_0, \\ 
    \bm P_a(z), & z\in U_a, \\ 
    \bm P_b(z), & z\in U_b.
\end{dcases}
\end{equation}

The final transformation then takes the form
\[
    \bm R(z) \deff \begin{dcases}
        \bm S(z) \bm P(z)^{-1}, \quad &z \in U\setminus \Gamma, \\
        \bm S(z) \bm G(z)^{-1}, &\text{elsewhere on }\C\setminus \Gamma_{\bm R}.
    \end{dcases}
\]

Both $\bm S$ and $\bm G$ have the same jumps on $(a,b)$. Likewise, both $\bm S$ and $\bm P_j$, $j=0,a,b$, have the same jumps inside $U_j$. These jumps cancel one another in the construction of $\bm R$, so that $\bm R$ is analytic across $[a,b]\cup (\Gamma_{\bm S} \cap U)$, and therefore it has jumps precisely across $\Gamma_{\bm R}$.

As a consequence, we obtain that $\bm R$ satisfies the following RHP.

  \begin{rhp}
        Seek for a matrix-valued function $\bm R : \C \setminus \Gamma_{\bm R} \to \C^{2 \times 2}$ satisfying the following.
        \begin{enumerate}[\rm (i)]
            \item $\bm R$ is analytic.
            \item $\bm R$ has continuous boundary values $\bm R_{\pm}$ along $\Gamma_{\bm R}$, that are related by $\bm R_+(z) = \bm R_-(z) \bm J_{\bm R}(z)$, where
            \begin{equation}\label{deff:JR}
                \bm J_{\bm R}(z) \deff
                    \begin{dcases}
                        \bm G(z) \bm J_{\bm S}(z)\bm G(z)^{-1}, & z\in \Gamma_{\bm R}\setminus \partial U, \\
                        \bm P(z)\bm G(z)^{-1}, & z\in \partial U.
                    \end{dcases}
            \end{equation}
            \item As $z \to \infty$, $\bm R$ behaves as 
            \[
                \bm R(z) = \bm I + \Boh \left(z^{-1} \right)
            \]
        \end{enumerate}
    \end{rhp}

To conclude the asymptotic analysis, we follow the usual small-norm theory path, and now prove that the jump matrix for $\bm R$ is asymptotically close to the identity matrix. We do it in the two separate lemmas that follow.

\begin{lemma}
    The estimate
    $$
    \|\bm J_{\bm R}-\bm I\|_{L^1\cap L^\infty (\partial U)}=\Boh(n^{-1}),\quad n\to\infty,
    $$
    is valid uniformly for $\sad\geq \sad_0$, for any $\sad_0\in \R$ fixed, and uniformly for $\tad$ in compacts of $(0,+\infty)$.
\end{lemma}
\begin{proof}
    The claim on the $L^\infty$ norm follows from the definition of $\bm P$ in \eqref{deff:PPoPaPb} and the asymptotics in \eqref{eq:precise_asymptotics_P_b}, \eqref{eq:precise_asymptotics_P_a} and RHP~\ref{rhp:localP0}. The claim on the $L^1$ norm then follows from the $L^\infty$ norm simply because $\partial U$ is a bounded set.
\end{proof}

\begin{lemma}
    Given $\sad_0\in \R$, there exists $\eta>0$ such that the estimate
    $$
    \|\bm J_{\bm R}-\bm I\|_{L^1\cap L^\infty (\Gamma_{\bm R}\setminus \partial U)}=\Boh(\ee^{-\eta n}),\quad n\to\infty,
    $$
    is valid uniformly for $\sad\geq \sad_0$, and also uniformly for $\tad$ in compacts of $(0,+\infty)$.
\end{lemma}
\begin{proof}
The claim follows from Proposition~\ref{prop:estimateJS} and the fact that $\bm G$ is bounded on compact subsets of $\C\setminus [a,b]$, we skip details.
\end{proof}

As a consequence, we finally conclude the small norm theory for $\bm R$.

\begin{theorem}\label{thm:smallnormROPs}
    The estimate
    $$
    \|\bm R-\bm I\|_{L^\infty(\C\setminus \Gamma_{\bm R})}=\Boh(n^{-1})\qquad \text{and}\qquad \|\bm R_\pm-\bm I\|_{L^2(\Gamma_{\bm R})}=\Boh(n^{-1}),\quad n\to \infty,
    $$
    are valid uniformly for $\sad\geq \sad_0$, for any $\sad_0\in \R$ fixed.
\end{theorem}

We now move to drawing the main conclusions of the asymptotic analysis.

\section{Consequences of the asymptotic analysis}\label{sec:spoils}

Having completed the asymptotic analysis of the RHP for OPs we prove our main results in this section.

\subsection{Asymptotics for the kernel: proof of Theorem~\ref{thm:CorrelationKernelAsymptotics}}\label{sec:KernelAsymptotics}\hfill 

We unwrap here all the transformations $\bm Y\mapsto \bm T\mapsto \bm S\mapsto \bm R$ to express $\bm Y_+$ in terms of the solution $\bm \Phi$ to the model problem.

For $z\in \R$ in a neighborhood of the origin, this unwrapping unravels the identity
\begin{equation}\label{eq:unwrapYkernel}
    \bm Y_+(z) 
    % & =\ee^{-n\ell/2\sp_3}\bm T_+(z)\ee^{n(g_+(z)+\ell/2)\sp_3}  =\ee^{-n\ell/2\sp_3}\bm S_+(z)\left(\bm I+\dfrac{\ee^{2n\phi_{b,+}(z)}}{\sigma_n(z)}\bm E_{21}\right)\ee^{n(\msf g_+(z)+\ell/2)\sp_3}    \\
    %  &  
     =\ee^{-n\ell/2\sp_3}\bm R_+(z)\bm P_{0,+}(z)\left(\bm I+\dfrac{\ee^{2n\phi_{b,+}(z)}}{\sigma_n(z)}\bm E_{21}\right)\ee^{n(\msf g_+(z)+\ell/2)\sp_3}.  
     \end{equation}

To lighten notation, let us introduce for a moment the unweighted version of the kernel, namely
\begin{equation}\label{eq:KnhatKn}
\wh{\msf K}_n(x,y)\deff \frac{2\pi \ii (x-y)}{\sqrt{\omega_n(x)}\sqrt{\omega_n(y)}}\msf K_n(x,y).
\end{equation}

For $x,y \in \R$ and near $z=0$, \eqref{eq:relKnRHPY} and \eqref{eq:unwrapYkernel} combined show that 
\begin{equation}\label{eq:esthatK1}
\wh{\msf K}_n(x,y)
    =\ee^{n(\msf g_+(x)+\msf g_+(y)+\ell)}
    \left[\left(\bm I-\frac{\ee^{2n\phi_{b}(y)}}{\sigma_n(y)}\bm E_{21}\right)\bm P_0(y)^{-1}\bm R(y)^{-1}\bm R(x)\bm P_0(x)\left(\bm I+\frac{\ee^{2n\phi_{b}(x)}}{\sigma_n(x)}\bm E_{21}\right)\right]_{21,+}. 
\end{equation}

Write
\begin{equation}\label{eq:Rtrick}
\bm R(y)^{-1}\bm R(x)
    % =\bm R(y)^{-1}\left(\bm R(x)-\bm R(y)+\bm R(y)\right)
    =\bm I+\bm R(y)^{-1}\left(\bm R(x)-\bm R(y)\right),
\end{equation}
Cauchy's integral formula and Theorem~\ref{thm:smallnormROPs} gives the estimate,
$$
\bm R(x)-\bm R(y)=\frac{1}{2\pi\ii}(x-y)\oint_\gamma\frac{\bm R(w)}{(w-x)(w-y)}\dd w=\Boh\left(\frac{x-y}{n}\right).
$$
where $\gamma$ is any positively oriented closed contour surrounding $x$  and $y$.
This way, Equation~\eqref{eq:esthatK1} becomes
\begin{multline*}
    \ee^{-n(\mathsf{g}_+(x)+\mathsf{g}_+(y)+\ell)}
    \wh{\mathsf{K}}_n\left(x,y\right)
    = \left[ \left( \bm I-\frac{\ee^{2n\phi_{b}(y)}}{\sigma_n(y)}\bm E_{21} \right) \bm P_0(y)^{-1} \bm P_0(x) \left( \bm I+\frac{\ee^{2n\phi_{b}(x)}}{\sigma_n(x)}\bm E_{21} \right) \right]_{21,+} \\
    + \left[ \left( \bm I-\frac{\ee^{2n\phi_{b}(y)}}{\sigma_n(y)}\bm E_{21} \right) \bm P_0(y)^{-1} \Boh\left(\frac{x-y}{n}\right) \bm P_0(x)\left( \bm I+\frac{\ee^{2n\phi_{b}(x)}}{\sigma_n(x)}\bm E_{21} \right) \right]_{21,+}.
\end{multline*}
By Lemma \ref{lem:bound_P0}, both $\bm P_{0,+}(x)$ and $\bm P_{0,+}(y)^{-1}$ remain bounded for $x,y$ real near the origin.
Moreover, $\phi_{b,+}$ is purely imaginary on $(-\infty,b)$, implying that $\ee^{2n\phi_b(x)}/\sigma_n(x)$ and $\ee^{2n\phi_b(x)}/\sigma_n(x)$ are bounded for $x$ and $y$ real close to $0$. Hence,
\[
    \left[ \left( \bm I-\frac{\ee^{2n\phi_b(y)}}{\sigma_n(y)}\bm E_{21} \right) \bm P_0(y)^{-1} \Boh\left(\frac{x-y}{n}\right) \bm P_0(x)\left( \bm I+\frac{\ee^{2n\phi_b(x)}}{\sigma_n(x)}\bm E_{21} \right) \right]_{21,+} = \Boh \left( \frac{x-y}{n}\right)
\]
where the $\Boh$ term is uniform for $x,y$ small.
By the definition of $\bm P_0$ (see \eqref{eq:def_P0}), the modified kernel can be written as
\begin{multline*}
    \ee^{-n(\mathsf{g}_+(x)+\mathsf{g}_+(y)+\ell)} \wh{\mathsf{K}}_n\left(x,y\right)
        =\ee^{n(\phi_{b,+}(x)+\phi_{b,+}(y))} \left[ \left( \bm I-\frac{\bm E_{21}}{\sigma_n(y)} \right) \bm \Phi_n(n\varphi_0(y))^{-1} \ee^{-\ii n\kappa\sp_3}\right. \\
    \left. \times  
     \bm M(y)^{-1} \bm M(x) \ee^{\ii n\kappa\sp_3} \bm \Phi_n(n\varphi_0(x)) \left( \bm I+\frac{\bm E_{21}}{\sigma_n(x)}\right) \right]_{21,+} + \Boh\left(\frac{x-y}{n}\right) .
    \end{multline*}

    A similar calculation to \eqref{eq:Rtrick} gives 
    \[
        \bm M(y)^{-1} \bm M(x) = \bm I + \Boh(x-y),
    \]
    and, using the expressions for $\phi_{b,+}$ and $\msf g_+$ in terms of $V$ in \eqref{eq:phibg} and \eqref{eq:phibV}, and the fact that $\phi_{b,+}$ is purely imaginary on $(-\infty, b)$, we rewrite above as
\begin{multline*}
     \ee^{ -\frac{n}{2} \left( V(x)+V(y)\right)} \wh{\msf K}_n(x,y) 
        = \left[ \left( \bm I-\frac{\bm E_{21}}{\sigma_n(y)} \right) \bm \Phi_n(n \varphi_0(y))^{-1} \bm \Phi_n(n \varphi_0(x)) \left( \bm I+\frac{\bm E_{21}}{\sigma_n(x)}\right) \right]_{21,+} 
        \\
        + \left[ \left( \bm I-\frac{\bm E_{21}}{\sigma_n(y)} \right) \bm \Phi_n(n \varphi_0(y))^{-1}   
        \Boh(x-y) \bm \Phi_n(n \varphi_0(x)) \left( \bm I+\frac{\bm E_{21}}{\sigma_n(x)}\right) \right]_{21,+} 
     + \Boh\left(\frac{x-y}{n}\right) .
\end{multline*}
Hence, moving back from $\wh{\msf K}_n$ to $\msf K_n$ (recall \eqref{eq:KnhatKn}),
\begin{multline*}
     \frac{{\msf K}_n(x,y)}{\sqrt{\sigma_n(x)} \sqrt{\sigma_n(y)}} 
        = \frac{1}{ 2 \pi \ii (x-y) } \left[ \left( \bm I-\frac{\bm E_{21}}{\sigma_n(y)} \right) \bm \Phi_n(n \varphi_0(y))^{-1} \bm \Phi_n(n \varphi_0(x)) \left( \bm I+\frac{\bm E_{21}}{\sigma_n(x)}\right) \right]_{21,+} \\
      \\ + \frac{1}{ 2 \pi \ii  } \left[ \left( \bm I-\frac{\bm E_{21}}{\sigma_n(y)} \right) \bm \Phi_n(n \varphi_0(y))^{-1} \Boh(1) \bm \Phi_n(n \varphi_0(x)) \left( \bm I+\frac{\bm E_{21}}{\sigma_n(x)}\right) \right]_{21,+} + \Boh\left(\frac{1}{n}\right) 
\end{multline*}
  
    Introduce local variables $\zeta,\xi$ by
    $$
    \zeta= n\varphi_0'(0)x=n\pi\phi_V(0)x \qquad \text{and}\qquad  \xi = n\varphi_0'(0)y=n\pi\phi_V(0)y,
    $$
where we used Proposition~\ref{prop:properties_of_varphi_0} for the last identities in each equation.
Observe that these new variables satisfy
$$
x = \varphi_0^{-1}\left(\frac{\zeta}{n}\right)+\Boh\left(\frac{1}{n^2}\right)\quad \text{and}\quad y=\varphi_0^{-1}\left(\frac{\xi}{n}\right)+\Boh\left(\frac{1}{n^2}\right),
$$
uniformly for $\zeta,\xi$ in compacts. Also
\[
    n \varphi_0(x) = \zeta + \Boh\left( \frac{1}{n}\right), \quad
    n^2 Q(x) = \uad \zeta^2 + \Boh\left(
    \frac{1}{n} \right),
\]
for $\zeta$ in compacts. Recall also that $\lambda_n$ was introduced in \eqref{eq:lambdanmodel0}.

With the notations we just discussed, we conclude
\begin{multline*}
\frac{1}{n\pi \phi_V(0)} \msf K_n\left( \frac{\zeta}{n \pi \phi_V(0)},\frac{\xi}{ n\pi \phi_V(0)} \right)\\ 
    = \frac{\sqrt{{\lambda}_n(\zeta)}\sqrt{{\lambda}_n(\xi)}}{2 \pi \ii (\zeta - \xi)} 
    \left[ \left( \bm I-\frac{\bm E_{21}}{{\lambda}_n(\xi)} \right) \bm \Phi_n(\xi)^{-1} \bm \Phi_n(\zeta) \left( \bm I+\frac{\bm E_{21}}{{\lambda}_n(\zeta)}\right)\right]_{21,+} + \Boh \left( \frac{1}{n}\right).
\end{multline*}

From the very definition of ${\lambda}_n(\zeta)=\sigma_n(x)$,
$$
    {\lambda}_n(\zeta) - \lambda_\infty(\zeta) = \Boh\left(\frac{\ee^{-\sad}}{n}\right),\quad \text{where we have set}\quad \lambda_\infty(\zeta) \defeq \frac{1}{1 + \ee^{- \sad -  \uad \zeta^2}},
$$
and where the error term is uniform for $\zeta$ in compacts and also uniform for $\sad\geq \sad_0$ for any $\sad_0\in \R$. 

Using this last estimate on $\lambda_n$ and Equation \eqref{eq:asymptotics_phi_n} we find the final asymptotic expression for $\msf K_n$, namely
\begin{multline*}
\frac{1}{n\pi \phi_V(0)} \msf K_n\left( \frac{\zeta}{n \pi \phi_V(0)},\frac{\xi}{ n\pi \phi_V(0)} \right)\\ 
        = \frac{\sqrt{\lambda_\infty(\zeta)}\sqrt{\lambda_\infty(\xi)}}{2 \pi \ii (\zeta - \xi)} 
    \left[ \left( \bm I-\frac{\bm E_{21}}{\lambda_\infty(\xi)} \right) \bm \Phi_\infty(\xi)^{-1} \bm \Phi_\infty(\zeta) \left( \bm I+\frac{\bm E_{21}}{\lambda_\infty(\zeta)}\right)\right]_{21,+} 
 + \Boh \left( \frac{1}{n^{1-\varepsilon}}\right).
\end{multline*}

Thanks to Equation \eqref{eq:Limiting_Kernel}, the limiting kernel above is precisely $\msf K_\infty$ as claimed in Theorem~\ref{thm:CorrelationKernelAsymptotics}. 

The nonlocal equation for $\msf \Phi$ is the same as \eqref{eq:nonlocalPDE}. Finally, the convergence \eqref{eq:sinekernellimit} is a consequence of \eqref{eq:Limiting_Kernel}, Theorem~\ref{thm:PhiInfinityConv} and \eqref{eq:SinekernelRHPcharac}. These considerations conclude the proof of Theorem~\ref{thm:CorrelationKernelAsymptotics}.

\subsection{Asymptotics for recurrence coefficients: proof of Theorem~\ref{thm:recurrence coeffs}}\hfill 

We now compute asymptotics for the recurrence coefficients $\gamma_n(\sad)$ and $\beta_n(\sad)$, proving Theorem \ref{thm:recurrence coeffs}.

Our starting point is the relation \eqref{eq:rec_coef_from_RHP}. Unwrapping the transformations $\bm Y\mapsto \bm T\mapsto \bm S\mapsto \bm R$ performed in the asymptotic analysis, we obtain the asymptotic expansion of $\bm Y$ as $z\to \infty$ in the form
    \[
        \begin{split}
        \bm Y(z)
            %&= \ee^{- n \frac{\ell}{2} \sp_3} \bm T(z) \ee^{n (\msf g(z) + \frac{\ell}{2 })\sp_3}, \\
            &= \ee^{- n \frac{\ell}{2} \sp_3} \bm R(z) \bm G(z) \ee^{n (\msf g(z) + \frac{\ell}{2 })\sp_3}, \\
            &= \ee^{- n \frac{\ell}{2} \sp_3} \left( \bm I + \frac{\bm R_1}{z} + \frac{\bm R_2}{z^2} + \Boh \left( \frac{1}{z^3}\right)\right) \left( \bm I + \frac{\bm G_1}{z} + \frac{\bm G_2}{z^2} + \Boh\left( \frac{1}{z^3}\right) \right) \ee^{n (\msf g(z) + \frac{\ell}{2 })\sp_3}.
        \end{split} 
    \]
    Thus
    \[
    \begin{split}
         \bm Y_1 
         &= -\res_{z = \infty} \left( \bm Y(z) z^{-n \sp_3} \right)  
       % &= -\res_{z = \infty} \left[ \ee^{- n \frac{\ell}{2} \sp_3} \left( \bm I + \frac{\bm R_1}{z}  + \cdots\right) \left( \bm I + \frac{\bm G_1}{z}  + \cdots \right) \left( \bm I + \frac{n \msf g_1\sp_3}{z}  +\cdots\right) \ee^{n \frac{\ell}{2 }\sp_3} \right], \\
        =  \ee^{- n \frac{\ell}{2} \sp_3}
 \left( \bm R_1 + \bm G_1 + n \msf g_1 \sp_3\right) \ee^{n \frac{\ell}{2} \sp_3},
 \end{split}
    \]
where we used the expansion
$$
\ee^{\msf g(z)\sp_3}=
z^{n\sp_3}\left( \bm I + \frac{n \msf g_1\sp_3}{z} ++ \frac{1}{z^2}\left(\frac{n^2}{2} \msf g_1^2\bm I+n\msf g_2\sp_3\right)+\Boh(z^{-3})\right).
$$

    Similar computations yield an expression for $\bm Y_2$, namely
    \[
    \begin{split}
         \bm Y_2 
        &= -\res_{z = \infty} \left( z\bm Y(z) z^{-n \sp_3} \right), \\ 
        &=  \ee^{- n \frac{\ell}{2} \sp_3}
 \left(\frac{n^2}{2} \msf g_1^2\bm I+n(\msf g_2+\msf g_1\bm G_1+\msf g_1\bm R_1)\sp_3+ \bm G_2 + \bm R_1\bm G_1+\bm R_2 \right) \ee^{n \frac{\ell}{2} \sp_3}.
  \end{split}
 \]
    Equations \eqref{eq:rec_coef_from_RHP} become
\begin{equation}\label{eq:gammanbetanexpansion}
    \gamma_n^2(\sad)=\left(\bm G_1+\bm R_1\right)_{12}\left(\bm G_1+\bm R_1\right)_{21}, \quad
    \beta_n(\sad)=
\frac{(\bm G_2+\bm R_1\bm G_1+\bm R_2)_{12}}{(\bm G_1+\bm R_1)_{12}}-(\bm G_1+\bm R_1)_{22}.
\end{equation}

    We start calculating $\bm G_1$ and $\bm G_2$. As $z \to \infty$,
    \[\begin{split}
        \bm G(z) 
        &= \ee^{- \msf h_0 \sp_3} \bm M (z) \ee^{\msf h(z) \sp_3} \\
        &= \ee^{- \msf h_0 \sp_3} \left( \bm I + \frac{\bm M_1}{z} + \frac{\bm M_2}{z^2} + \Boh\left( \frac{1}{z^3}\right)\right) \ee^{\msf h_0 \sp_3} \left( \bm I + \frac{\msf h_1 \sp_3 }{z} + \frac{ \frac{\msf h_1^2}{2} \bm I + \msf h_2 \bm \sigma_3 }{z^2} + \Boh \left( \frac{1}{z^3}\right)\right),
    \end{split}\]
    where $\msf h_0$ and $\bm M_1,\bm M_2$ are as in \eqref{eq:estimates_h0_h1} and \eqref{eq:asymptotics_M}, respectively.
    
    Therefore, introducing
    \begin{equation}\label{deff:hatG1hatG2}
    \begin{aligned}
        \widehat{\bm G}_1 &\defeq \widehat{\msf h}_0\, [\bm M_1, \sp_3] + \widehat{\msf h}_1 \sp_3 = -\frac{\ii (b-a)\wh{\msf h}_0}{2}\sp_1 + \wh{\msf h}_1\sp_3, \\
        \widehat{\bm G}_2 &\defeq \widehat{\msf h}_2 \sp_3 - \widehat{\msf h}_1 \frac{b-a}{4} \sp_2 +  \frac{\ii(b^2-a^2)\widehat{\msf h}_0}{4}  \sp_1,
    \end{aligned}
    \end{equation}
    the expansions
    \begin{equation}\label{deff:G1G2}
    \begin{aligned}
        \bm G_1 
            &= \ee^{- \msf h_0 \sp_3} \bm M_1 \ee^{\msf h_0 \sp_3} +  \msf h_1 \sp_3 
            = \bm M_1 + \frac{\widehat{\bm G}_1}{n} + \Boh(n^{-2}), \\
        \bm G_2 
            &= \frac{\msf h_1^2}{2}\bm I+\msf h_2\sp_3+\msf h_1\ee^{-\msf h_0\sp_3}\bm M_1\ee^{\msf h_0\sp_3}+\ee^{-\msf h_0\sp_3}\bm M_2\ee^{\msf h_0\sp_3} 
            = \bm M_2 + \frac{\widehat{\bm G}_2}{n} + \Boh(n^{-2}),
    \end{aligned}
    \end{equation}
    are valid.
  
 We now compute $\bm R_1$ and $\bm R_2$. The identity
    \[
        \bm R(z) = \bm I + \frac{1}{2 \pi \ii} \int_{\Gamma_\bm R} \bm R_-(w) \left( \bm J_{\bm R}(w) - \bm I \right) \frac{\dd w}{w-z}, 
    \]
    implies the expression
    \[
        \bm R_k 
            = - \frac{1}{2 \pi \ii} \int_{\Gamma_\bm R} \bm R_-(w) \left( \bm J_{\bm R}(w) - \bm I \right) w^{k-1} \dd w 
            = - \frac{1}{2 \pi \ii} \ointclockwise_{\partial U} (\bm J_{\bm R}(w) - \bm I) w^{k-1}\, \dd w 
            \left( \bm I +  \Boh \left( \frac{1}{n}\right) \right),
    \]    
    where we recall that $U=U_a\cup U_b\cup U_0$.
    To compute the integral above we need to study the local behavior of the integrand around the points $0$, $a$ and $b$. Using \eqref{eq:gpMGasympt}, 
    it is straightforward to check that
    \[
        \bm J_{\bm R}(z) - \bm I = \frac{\bm M(z)\bm J_1^{(n)}(z)\bm M(z)^{-1}}{n} 
        % + \frac{\bm J_2^{(n)}(z)}{n^2} 
        + \Boh(n^{-2}),
    \]
    where
    \begin{equation}\label{eq:J1n}
        \bm J_1^{(n)}(z) \defeq 
            \begin{dcases}
                \frac{1}{\varphi_0(z)}\ee^{\pm \ii n\kappa \sp_3}\left(\bm U^{\pm} \right)^{-1} \bm \Phi_{n,1} \bm U^{\pm} \ee^{\mp \ii n\kappa \sp_3} -\widehat{\msf h}(z) \sp_3 , \quad 
                &z \in U_0, \pm \im z>0, \\
                  \frac{\sp_3 \bm A_1 \sp_3 }{ \varphi_a(z)^{\frac{3}{2}}}, \quad 
                &z \in U_a, \\
                 \frac{\bm A_1 }{ \varphi_b(z)^{\frac{3}{2}}}, \quad 
                &z \in U_b,
            \end{dcases}
    \end{equation}
    and where $\bm A_1$ is as in RHP~\ref{rhp:airy}.
    The factor $\wh{\msf h}$ appearing above is the same as in Proposition~\ref{prop:h_n_asymptotics}. Observe that $\bm J_1^{(n)}$ still depends on $n$, but solely through the matrix $\bm \Phi_{n,1}$ and the oscillatory factors $\ee^{\mp \ii n\kappa \sp_3}$.

From these computations we deduce
\begin{equation}\label{eq:expansionRkIk}
    \bm R_k=
         \frac{\wh{\bm R}_k}{n}  + \Boh\left(\frac{1}{n^2}\right),
\end{equation}
with
\begin{equation}\label{deff:Ikp}
\wh{\bm R}_k\deff \msf I_0^{(k)}+\msf I_a^{(k)}+\msf I_b^{(k)},\quad \msf I_p^{(k)}\deff \frac{1}{2\pi \ii}\ointctrclockwise_{\partial U_p} \bm M(w)\bm J_1^{(n)}(w)\bm M(w)^{-1} w^{k-1}\dd w,\quad p=0,a,b, \; k=0,1.
\end{equation}

\begin{remark}
    The factors $\msf I_a^{(k)}$ and $\msf I_b^{(k)}$ yield contributions coming from the regular soft edges $a,b$, which in turn are contributions coming from Airy parametrices. These contributions are exactly the same that occur for unperturbed weights, that is, when we make $\sigma_n\equiv 1$. Their calculation have appeared before in the literature, although perhaps not so explicitly. For completeness, we now evaluate these contributions step by step.
\end{remark}

We now compute each of these integrals to leading order in $n$. Using \eqref{eq:gpMGasympt}, we expand
    \begin{align}
        \msf I_b^{(k)}
        &=\frac{1}{2\pi\ii}\ointctrclockwise \frac{1}{\varphi_b(w)^{3/2}}\bm U_0\left(\frac{w-b}{w-a}\right)^{\sp_3/4}\bm U_0^{-1}\bm A_1\bm U_0\left(\frac{w-b}{w-a}\right)^{-\sp_3/4}\bm U_0^{-1}w^{k-1}\dd w \nonumber \\
        %
        % &=\frac{1}{2\pi\ii} \ointctrclockwise \frac{1}{\varphi_b(w)^{3/2}}\bm U_0\left(\frac{w-b}{w-a}\right)^{\sp_3/4}\left[(\ii \alpha_1+\beta_1)\bm E_{12}+(-\ii \alpha_1+\beta_1)\bm E_{21}\right]\left(\frac{w-b}{w-a}\right)^{-\sp_3/4}\bm U_0^{-1}w^{k-1}\dd w\\
        %
        &=\frac{\ii \msf a_1}{2\pi\ii}\bm U_0\left[\ointctrclockwise \frac{1}{\varphi_b(w)^{3/2}}\left[5\left(\frac{w-b}{w-a}\right)^{1/2}\bm E_{12}+7\left(\frac{w-a}{w-b}\right)^{1/2}\bm E_{21}\right]w^{k-1}\dd w\right]\bm U_0^{-1}\nonumber \\
        &=   \frac{\msf a_1}{2}
        \left[5\res_{w=b}\left(\frac{w^{k-1}}{\varphi_b(w)^{3/2}}\left(\frac{w-b}{w-a}\right)^{1/2}\right) \left( \sp_3 + \ii \sp_1 \right)
        + 7 \res_{w=b}\left(\frac{w^{k-1}}{\varphi_b(w)^{3/2}}\left(\frac{w-a}{w-b}\right)^{1/2}\right) \left( - \sp_3 + \ii \sp_1 \right)\right]. \label{eq:Ikbresidue} 
    \end{align}

    The calculation of $\msf I_a^{(k)}$ goes on in a similar way, and we obtain
    \begin{equation}\label{eq:Ikaresidue}
        \msf I_a^{(k)}
            =-\frac{\msf a_1}{2}\left[ 7 \res_{w=a}\left(\frac{w^{k-1}}{\varphi_a(w)^{3/2}}\left(\frac{w-b}{w-a}\right)^{1/2}\right) \left( \sp_3 + \ii \sp_1\right)
            + 5 \res_{w=a}\left(\frac{w^{k-1}}{\varphi_a(w)^{3/2}}\left(\frac{w-a}{w-b}\right)^{1/2}\right)\left( - \sp_3 + \ii \sp_1\right)\right].
    \end{equation}

We now have to compute several residues at $p=a,b$, and their structure is the following. Recall that $\varphi_p=(\frac{3}{2}\phi_p)^{3/2}$, where $\phi_p$ is determined from \eqref{eq:phi_b}, \eqref{eq:phi_a} and \eqref{eq:CauchyTransform_AlgebraicEq}. Using \eqref{eq:CauchyTransform_AlgebraicEq}, we expand
\begin{align*}
    & \phi_a(z)
        =\frac{2}{3}q(a)\sqrt{b-a}(a-z)^{3/2}+\frac{1}{5}\left( \frac{q(a)}{\sqrt{b-a}}+2q'(a)\sqrt{b-a} \right)(a-z)^{5/2}+\Boh(|a-z|^{7/2}),\quad z\to a,\\
    & \phi_b(z)
        =\frac{2}{3}q(b)\sqrt{b-a}(z-b)^{3/2}+\frac{1}{5}\left( \frac{q(b)}{\sqrt{b-a}}+2 q'(b)\sqrt{b-a} \right)(z-b)^{5/2}+\Boh(|b-z|^{7/2}),\quad z\to b,
\end{align*}
where, as usual, all the roots above are with principal branch.

Expanding now $\varphi_p(z)$, we get
$$
\varphi_a(z)=-C_a(z-a)\Psi_a(z), \quad \varphi_b(z)=C_b(z-b)\Psi_b(z),\quad C_p\deff q(p)^{2/3}\sqrt[3]{b-a},\quad p=a,b,
$$
where $\Psi_p$ is analytic in a neighborhood of $p=a,b$, and satisfies
$$
    \Psi_p(p)=1,\quad \Psi'_a(a)=-\frac{1}{5}\frac{1}{b-a}-\frac{2}{5}\frac{q'(a)}{q(a)},\quad \Psi'_b(b)=\frac{1}{5}\frac{1}{b-a}+\frac{2}{5}\frac{q'(b)}{q(b)}.
$$

With these expansions, we now compute the residues appearing in \eqref{eq:Ikbresidue} and \eqref{eq:Ikaresidue}. They shall be given in terms of
\[
    \rho_1^{(k)}(p) \defeq   \frac{p^{k-1} \sign(p)}{q(p) } \frac{1}{b-a}, \quad 
    \rho_2^{(k)}(p) \defeq \frac{\rho_1^{(k)}(p)}{2} + \frac{p^{k-1}}{q(p)} \left[ \frac{
    k-1}{p} - \frac{3}{2} \Psi'_p(p)\right].
\]

The expressions for $\msf I_b$ and $\msf I_a$ updates to 
\begin{equation}\label{eq:contributions_at_a_and_b}
    \begin{split}
        \msf I_b^{(k)} 
        &= \frac{\msf a_1}{2}\left[ \left(  5 \rho_1^{(k)}(b) - 7\rho_2^{(k)}(b) \right) \sp_3 + \ii \left( 5 \rho_1^{(k)}(b) + 7 \rho_2^{(k)}(b) \right) \sp_1 \right], \\
    \msf I_a^{(k)}
        &= \frac{\msf a_1}{2} \left[ \left(  5 \rho_1^{(k)}(a) - 7\rho_2^{(k)}(a)\right) \sp_3 - \ii \left( 5 \rho_1^{(k)}(a) + 7\rho_2^{(k)}(a)  \right) \sp_1 \right].    
    \end{split}
\end{equation}

We now compute $\msf I_0^{(k)}$. From \eqref{eq:J1n}, \eqref{deff:Ikp} and the explicit expression for $\bm M$ in \eqref{eq:M},
    \begin{equation}\label{eq:contribution_at_0}
        \begin{multlined}
            \msf I_0^{(k)} 
                = \bm U_0 \left[ \frac{1}{2 \pi \ii} \ointctrclockwise_{\partial D_0} \left( \frac{w-b}{w-a}\right)^{\frac{\sp_3}{4}} \bm U_0^{-1}  \frac{1}{\varphi_0(w)} \ee^{\pm \ii n \kappa \sp_3} \left( \bm U ^{\pm}\right)^{-1} \bm \Phi_{n,1} \bm U^{\pm } \ee^{\mp \ii n \kappa \sp_3} \bm U_0 \left( \frac{w -b}{w-a}\right)^{-\frac{\sp_3}{4}}  w^{k-1} \, \dd w \right. \\
                    \left. - \frac{1}{2 \pi \ii} \ointctrclockwise_{\partial D_0} \widehat{\msf h}(w) \left( \frac{w-b}{w-a}\right)^{\frac{\sp_3}{4}} \bm U_0^{-1}   \sp_3  \bm U_0 \left( \frac{w -b}{w-a}\right)^{-\frac{\sp_3}{4}} w^{k-1} \, \dd w \right] \bm U_0^{-1},
        \end{multlined}
    \end{equation}
    where we emphasize that the choices $\pm =+$ and $\mp=-$ are taken in the part of the contour on the upper half plane, and $\pm=-$ and $\mp=+$ are taken in the lower half plane. With these choices, a straightforward calculation shows that the correspond integrand is in fact analytic (for $k=2$) or meromorphic (for $k=1$) in a full neighborhood of the origin, with a sole pole, which is simple, at the origin. 
    
    Therefore, the integrals in \eqref{eq:contribution_at_0} reduce to residue calculations, in which we may (and will) use the $+$-boundary values of the expressions, and we obtain
    \[
    \msf I_0^{(2)}=0,
    \]
    and writing $\varphi_0(z) = \pi \phi_V(0) z + \Boh(z^{-2})$ as $z \to 0$ (see Proposition \ref{prop:properties_of_varphi_0}),
        \begin{align*}
            \msf I_0^{(1)} 
             & = \bm U_0 \left( \frac{b}{a} \right)^{\frac{\sp_3}{4}}_+ \bm U_0^{-1} \left[  \frac{1}{\pi \phi_V(0)} \ee^{ \ii n \kappa \sp_3} \bm \Phi_{n,1} \ee^{-\ii n \kappa \sp_3} - \left( \res_{z = 0} \widehat{\msf h}(z) \sp_3 \right)  \right] \bm U_0\left( \frac{b}{a} \right)^{-\frac{\sp_3}{4}}_+ \bm U_0^{-1} \\
             & = \bm U_0 \left( \frac{b}{a} \right)^{\frac{\sp_3}{4}}_+ \bm U_0^{-1} \left[  \frac{1}{\pi \phi_V(0)} \ee^{\ii n \kappa \sp_3} \bm \Phi_{\infty,1} \ee^{-\ii n \kappa \sp_3} - \left( \res_{z = 0} \widehat{\msf h}(z) \sp_3 \right)  \right] \bm U_0\left( \frac{b}{a} \right)^{-\frac{\sp_3}{4}}_+ \bm U_0^{-1} +\Boh(n^{-1+\varepsilon}),
        \end{align*}
        where we used \eqref{eq:asymptotics_phi_n} for $0 < \varepsilon < 1$.
        
    From the very definition of $\widehat{\msf h}(z)$ (see Proposition \ref{prop:h_n_asymptotics}),
    \[
        \res_{z = 0} \widehat{\msf h}(z) = \ii \widehat{\msf h}_0 \sqrt{-ab} = \ii \frac{G_0(\sad)}{2 \pi \sqrt{\msf t}}.
    \]

    Recalling \eqref{eq:Phi1_decomposition}, we express $\bm \Phi_{\infty,1}$ in terms of $\msf p,\msf q$. After a rather cumbersome but straightforward calculation, we obtain the final expression
    \begin{multline*}
        \msf I_0^{(1)} 
            = \frac{\msf p(\sad)}{2 \pi \phi_V(0)} \left( \frac{b+a}{\sqrt{-ab}} \sp_3 + \ii\frac{b-a}{\sqrt{-ab}} \sp_1 \right) 
        +  \frac{G_0(\sad)}{4 \pi\sqrt{\msf t}} \left( \frac{b+a}{\sqrt{-ab}} \sp_3 + \ii\frac{b-a}{\sqrt{-ab}} \sp_1 \right) \\
        +  \frac{\msf q(\sad)}{2 \pi \phi_V(0)} \left( -2\cos(2n\kappa) \sp_2 - \frac{b-a}{\sqrt{-ab}}\sin(2n\kappa) \sp_3 + \ii\frac{b+a}{\sqrt{-ab}}\sin(2n\kappa) \sp_1 \right).
    \end{multline*}

    We now go back to the recurrence coefficients in \eqref{eq:gammanbetanexpansion}. Starting with $\gamma_n^2(\sad)$, we evaluate 
    \[
    \begin{split}
        \gamma_n^2(\sad) 
            % &= \left( \bm G_1 + \bm R_1 \right)_{12} \left( \bm G_1 + \bm R_1 \right)_{21} \\
            % %
            &= \left( \bm M_1\right)_{12}  \left( \bm M_1 \right)_{21} + \frac{1}{n} \left[ \left( \bm M_1\right)_{12} \left( \wh{\bm G}_1\ + \wh{\bm R}_1 \right)_{21} + \left( \bm M_1\right)_{21} \left( \wh{\bm G}_1 + \wh{\bm R}_1\right)_{12} \right] + \Boh(n^{- 2}) \\
            &= \frac{(b-a)^2}{16} + \frac{1}{n }\left[ \ii \frac{b-a}{4} \left(\left( \wh{\bm R}_1\right)_{21} - \left( \wh{\bm R}_1\right)_{12} \right) \right] + \Boh(n^{- 2}) \\
            % %
            &= \frac{(b-a)^2}{16} +\frac{\cos(2 n \kappa)}{n} \frac{b-a}{2} \frac{\msf q(\sad)}{\pi \phi_V(0)}  + \Boh(n^{ - 2+\varepsilon}).
    \end{split}
    \]
    where we used the symmetries $(\bm M_1)_{12}=-(\bm M_1)_{21}$, and $(\bm B )_{12}=(\bm B )_{21}$ for $\bm B =\wh{\bm G}_1, \msf I^{(k)}_a, \msf I^{(k)}_b$, see \eqref{deff:M1M2}, \eqref{deff:hatG1hatG2} and \eqref{eq:contributions_at_a_and_b}.

    For $\beta_n(\sad)$ we use \eqref{deff:G1G2} to compute
\[
    \begin{split}
        \beta_n(\sad) 
            % &= \frac{\left(\bm G_2 + \bm R_1 \bm G_1 + \bm R_2 \right)_{12}}{\left( \bm G_1 + \bm R_1 \right)_{12}} - \left( \bm G_1 + \bm R_1 \right)_{22} \\
            %%
            &= \frac{\left( \bm M_2\right)_{12}}{\left( \bm M_1\right)_{12}} - \left( \bm M_1\right)_{22} \\ 
            &\phantom{=} + \frac{1}{n} \left[ 
            - \frac{\left( \bm M_2\right)_{12}}{\left( \bm M_1\right)^2_{12}} \left( \wh{\bm G}_1 + \wh{\bm R}_1\right)_{12} 
            + \frac{1}{\left( \bm M_1\right)_{12}} \left( \wh{\bm G}_2 + \wh{\bm R}_2 + \wh{\bm R}_1 \bm M_1 \right)_{12} 
            - \left( \wh{\bm G}_1 + \wh{\bm R}_1 \right)_{22} 
            \right] + \Boh(n^{-2}). \\
            &= \frac{b+a}{2} +  \frac{1}{2n(b-a)} \left(\frac{1}{q(b)} - \frac{1}{q(a)} \right) \\
            &\phantom{=} + 
            \frac{1}{n} 
            \left[ 
                 \frac{(a+b)}{\sqrt{-ab}} \frac{G_0(\sad)}{2 \pi \sqrt{t}} - \frac{2 \msf q(\sad)}{\pi \phi_V(0)} \frac{1}{b-a} \left[ (a+b) \cos(2n\kappa) + 2 \sqrt{-ab} \sin(2n\kappa)\right]
                \right] + \Boh(n^{-2 + \varepsilon}).
    \end{split}
\]
where we remind that $\kappa \defeq \pi \mu_V(0,b)$. Having in mind the transformation between $\msf q$ and $\msf Q$ from \eqref{eq:msfqpmsfPQ} and the explicit expressions for $\beta_n(\infty),\gamma_n^2(\infty)$ from \eqref{deff:gammabetainfty}, the asymptotic expansions for $\beta_n(\sad),\gamma_n^2(\sad)$ claimed in Theorem~\ref{thm:recurrence coeffs} follow from the definition of $\msf T$. 

The characterization of $\msf Q$ as in \eqref{eq:integralreprQ} was already explained in Section~\ref{sec:modelprob_particularcase}, see Remark~\ref{remark:ClaeysTarricone} and the comments thereafter. 

The decay \eqref{eq:decayG0q} is a consequence of \eqref{eq:estimates_pq} and \eqref{eq:msfqpmsfPQ}.

The proof of Theorem~\ref{thm:recurrence coeffs} is thus complete.
    
\appendix

\section{Laplace-Type Integrals}\label{Ap:LaplaceIntegrals}

Motivated by the understanding of the asymptotics of $\msf h$ as $z,n \to \infty$ in Section \ref{sec:RHPAnalysis}, we consider integrals of the type
\[
    F(t) \defeq \int_a^b g(x) \log(1 + \ee^{-y - t f(x)}) \, \dd x, \quad a < 0 < b.
\]
More precisely, we are interested in the asymptotics  of $F(t)$ as $t \to +\infty$.

    We start by considering the case $f(x) = x^2$.

\begin{lemma}\label{lem:Laplace}
    Let $g$ be a $L^1([a,b])$ function that is $C^\infty$ in a neighborhood of $0$ and $f(x) = x^2$. Assume that
    \begin{equation}
        y > -M,
    \end{equation}
    for some $M > 0$ fixed. For any $N>0$, the function $F(t)$ admits an expansion of the form
    \begin{equation}\label{eq:LaplaceExpformula}
        F(t) = \frac{1}{t^{\frac{1}{2} }} \left[ \sum_{k=0}^{2N } \frac{g^{(2k)}(0)}{(2k)!} \frac{G_{2k}(y)}{t^{k}} + \Boh \left( \frac{1}{t^{2N+1}}\right) \right], \quad t \to +\infty,
    \end{equation}
    where the coefficients $G_\beta(y)$ are given by
    \[
        G_\beta(y) = \int_{- \infty}^\infty u^\beta \log \left( 1 + \ee^{-y - u^2} \right) \, \dd u, \quad \beta \geq 0,
    \]
    and the $\Boh$ term is valid uniformly for $y \geq -M$.
\end{lemma}

\begin{remark} \hfill
\begin{itemize}
    \item Observe that $G_\beta(2k+1) = 0$ because the integrand in the definition of $G_\beta$ is an odd function, explaining why only even indices $G_{2k}$ appear in the expansion \eqref{eq:LaplaceExpformula}.
    \item Assuming that $\beta$ is even, we can write $G_\beta(y)$ in terms of the \textit{polylog function} $\text{Li}_{\gamma}$:
    \[
        G_\beta(y) 
        % = F_{\frac{\beta - 1}{2}}(y) 
        = - \frac{\beta - 1}{2} \Gamma \left( \frac{\beta - 1}{2}\right) \text{Li}_{2 + \frac{\beta - 1}{2}} \left(- \ee^{-y}\right).
    \]
\end{itemize}
    
\end{remark}
\begin{proof}
    For every $\delta > 0$, write
    \[
        \int_{a}^b g(x) \log\left(1 + \ee^{- y - t x^2} \right) \, \dd x = 
        \left( \int_{a}^{-\delta} + \int_{-\delta}^\delta + \int_\delta^b  \right) g(x)  \log\left(1 + \ee^{- y - t x^2} \right) \, \dd x.
    \]
    Note that
    \[
       \int_{[a, -\delta] \cup [\delta, b]} g(x) \log\left(1 + \ee^{- y - t x^2} \right) \, \dd x  \leq 2 \| g \|_1 \ee^M \ee^{-t \frac{\delta^2}{4}}.
    \]
    Thus, for some $c = c(\delta)> 0$,
    \begin{equation}\label{eq:appendix1_estimate1}
        \int_{a}^b g(x) \log\left(1 + \ee^{- y - t x^2} \right) \, \dd x = 
        \int_{-\delta}^\delta g(x) \log\left(1 + \ee^{- y - t x^2} \right) \, \dd x +  \Boh \left( \ee^{- c t}\right),
    \end{equation}
    where the $\Boh$ term is uniform for $y > -M$.
    Now fix $\delta > 0$ for which $g$ admits a Taylor expansion of order $N$ at the interval $(-\delta, \delta)$. Then
    \[
        \int_{-\delta}^\delta g(x) \log\left(1 + \ee^{- y - t x^2} \right) \, \dd x 
        = \sum_{j=0}^N \frac{g^{(j)}(0)}{j!} \int_{-\delta}^\delta x^j \log\left(1 + \ee^{- y - t x^2} \right) \, \dd x + R_{N+1}(t)
    \]
    where the remainder $R_{n+1}$ satisfies
    \[
        \left| R_{N+1}(t) \right| \leq \frac{1}{(N+1)!} \sup_{- \delta \leq \xi \leq \delta} \left| g^{(N+1)}(\xi) \right| \left| \int_{-\delta}^{\delta} x^{N+1} \log\left(1 + \ee^{- y - t x^2} \right) \, \dd x \right|.
    \]

    For $\beta \in \Z_+$, exchange variables $u = \sqrt{t} \cdot x$ to obtain
    \[
        \int_{- \delta}^\delta x^\beta \log (1 + \ee^{- y - t x^2}) \, \dd x = \frac{1}{t^{\frac{\beta + 1}{2}}} \left[ G_\beta(y)  - 
        \int_{[-\infty, - \delta \sqrt{ t} ] \cup [\delta \sqrt{ t}, +\infty]} u^\beta \log ( 1 + \ee^{ - y - u^2}) \, \dd u \right].
    \]
    Since
    \[
        \int_{[-\infty, - \delta \sqrt{ t} ] \cup [\delta \sqrt{ t}, +\infty]} u^\beta \log ( 1 + \ee^{ - y - u^2}) \, \dd u 
        \leq 2 \left\| u^\beta \ee^{- \frac{u^2}{2}} \right\|_1 \ee^M \ee^{- \frac{\delta^2 t}{2}}
    \]
    we get, for some $\widetilde{c} = \widetilde{c}(\delta) > 0$,
    \[
        \int_{- \delta}^\delta x^\beta \log (1 + \ee^{- y - t x^2}) \, \dd x = \frac{1}{t^{\frac{\beta + 1}{2}}} F_\beta(y) + \Boh(\ee^{- t \widetilde{c}}).
    \]
    where, once again, the $\Boh$ term is uniform for $y > -M$. Finally, plugging above identities on \eqref{eq:appendix1_estimate1} we obtain
    \[
        \begin{split}
            F(t) 
                &= \int_{-\delta}^\delta g(x) \log\left(1 + \ee^{- y - t x^2} \right) \, \dd x +  \Boh \left( \ee^{- c t}\right) \\
                &= \sum_{j=0}^N \frac{g^{(j)}(0)}{j!}\frac{G_j(y)}{t^{\frac{j+1}{2}}}
                + \frac{1}{t^{\frac{N+2}{2}}} \frac{\sup_{-\delta < \xi < \delta}|g^{(N+1)}(\xi)|}{(N+1)!} G_{N+1}(y)
                + \Boh(\ee^{-t \widetilde{c}}) + \Boh \left( \ee^{- c t}\right) \\
                &= \sum_{j=0}^N \frac{g^{(j)}(0)}{j!} \frac{G_j(y)}{t^{\frac{j+1}{2}}}  + \Boh \left( \frac{1}{t^{\frac{N+2}{2}}} \right) .
        \end{split}
    \]
    Replacing $N\mapsto 2N$, the proof is complete.
\end{proof}

\begin{prop}\label{prop:Laplace}
    Assume that $g$ is as in Lemma \ref{lem:Laplace} and that $f$ is $C^\infty$ in a neighborhood of the origin with a unique global minimum on $[a,b]$ at $x =0$, with $f(0) = f'(0) = 0$ and $f''(0) \neq 0$. Assume also that $ y > -M$ for some $M > 0$ fixed. Then $F(t)$ admits an expansion of the form
    \[
        F(t) =\frac{1}{t^{1/2}} \left(\sum_{k=0}^{2N} \frac{\widehat{g}^{(2k)}(0)}{(2k)!} \frac{G_{2k}(y)}{t^{k}} + \Boh \left( \frac{1}{t^{2N+1}}\right)\right), \quad t \to +\infty,
    \]
    where $\widehat g$ is a function that is $C^\infty$ in a neighborhood of $0$ whose derivatives at $0$ are given in terms of the derivatives of $f$ and $g$ at $0$. Moreover the $\Boh$ term is uniform for $y \geq -M$.
\end{prop}

\begin{proof}
    Consider the two-variable function
    \[
        H(u,v) = \frac{f(uv)}{u^2} - 1 
    \]
    Then, for $v_0 = \sqrt{\frac{2}{f''(0)}}$,
    \[
        \frac{\partial H}{\partial v} (0,v_0) = f''(0) \cdot v_0 \neq 0.
    \]
    Since $H(0,v_0) = 0$, the Implicit Function Theorem implies that there exists a function $v:[-\delta,\delta] \to \R$ such that 
    \[
        H(u,v(u)) = 0, \quad u \in [-\delta,\delta].
    \]
    The identity $f(u \cdot v(u)) = u^2$ implies that $v^{(k)}(0)$ is given in terms of $f^{(j)}(0)$, $0 \leq j \leq k+2$.

    Consider
    \[
        F(t) = \int_{-\delta}^{\delta} g(x) \log\left( 1+\ee^{-y-tf(x)} \right) \, \dd x + \int_{[a,-\delta] \cup [\delta, b]}  g(x) \log \left(1+\ee^{-y-t f(x)} \right)\, \dd x.
    \]
    For the first term in the r.h.s. the change of variables $x = u \cdot v(u)$ gives
    \[
        \int_{a_\delta}^{b_\delta} g(u \cdot v(u)) [v(u) + u \cdot v'(u)] \log \left( 1 + \ee^{- y - t u^2} \right) \, \dd u, \quad 
        a_\delta \defeq - \sqrt{f(-\delta}), \quad 
        b_\delta \defeq \sqrt{f(\delta)}.
    \]
    The result follows by applying Lemma \ref{lem:Laplace} to above integral, where $\widehat{g}(u) = g(u \cdot v(u))[ v(u) + u \cdot v'(u)]$. Differentiating above identity at $u=0$ one verifies that  $\widehat{g}^{(j)}(0)$ is given in terms of $g^{(k)}(0)$ and $f^{(l)}(0)$, $0 \leq k \leq j$, $0 \leq l \leq j+2$.
\end{proof}

  \bibliographystyle{abbrv}
\bibliography{Main.bib}
\end{document}